\documentclass{article}
\usepackage[utf8]{inputenc}

\usepackage{PRIMEarxiv}
\usepackage{amssymb,amsfonts,amsbsy,amsmath,amsthm}
\usepackage{graphicx}
\usepackage{float}
\usepackage{epsfig}
\usepackage{multirow}
\usepackage{color}
\usepackage{colortbl}
\usepackage{ulem}
\usepackage{bm}
\usepackage{comment}
\usepackage[]{algorithm2e}
\usepackage{mathtools}
\usepackage[title]{appendix}
\usepackage{xurl}
\usepackage{caption}
\usepackage{subcaption}
\usepackage{placeins} 

\newtheorem{theorem}{Theorem}[section]

\newtheorem{lemma}[theorem]{Lemma}
\newtheorem{example}{Example}
\newtheorem*{remark}{Remark}

\DeclareMathOperator{\E}{\mathbb{E}}
\DeclareMathOperator{\Cov}{Cov}
\DeclareMathOperator{\corr}{Corr}

\DeclareMathOperator{\vecc}{\text{vecc}}
\newcommand\norm[1]{\left\lVert#1\right\rVert}
\newcommand{\matr}[1]{\mathbf{#1}}
\renewcommand{\vec}[1]{\boldsymbol{#1}}
\DeclareMathOperator{\md2}{\mbox{MD}^2}

\DeclareMathOperator{\Tr}{Tr}

\DeclareMathOperator{\med}{median}
\DeclareMathOperator{\diag}{diag}
\DeclareMathOperator{\clr}{clr}
\DeclareMathOperator{\ilr}{ilr}

\title{Edgewise outliers of network indexed signals
}

\author{Christopher Rieser \\
  Institute of Statistics and Mathematical Methods in Economics\\
TU Wien \\
  Wiedner Hauptstraße, 1040 Vienna, Austria \\
  \\
  \texttt{christopher.rieser@tuwien.ac.at} \\
  \And
  Anne Ruiz-Gazen\\
  TSE Research Faculty \\
  Toulouse School of Economics \\
  1, Esplanade de l'Université
31080 Toulouse, France \\
  \texttt{anne.ruiz-gazen@tse-fr.eu} \\
  \And 
  Christine Thomas-Agnan \\
    TSE Research Faculty\\
  Toulouse School of Economics \\
  1, Esplanade de l'Université
  31080 Toulouse, France \\
  \texttt{christine.thomas@tse-fr.eu} \\
}

\begin{document}
\maketitle

\begin{abstract}
    We consider models for network indexed multivariate data involving a dependence between variables as well as across graph nodes. 
    In the framework of these models, we focus on  outliers detection and introduce the concept of edgewise outliers. For this purpose, we first derive the distribution of some sums of squares, in particular squared Mahalanobis distances that can be used to fix detection rules and thresholds for outlier detection. We then propose a robust version of the deterministic MCD algorithm that we call edgewise MCD.  An application on simulated data shows the interest of taking the dependence structure into account. We also illustrate the utility of the proposed method with a real data set. 
\end{abstract}

\section{Introduction}

Many real-world multivariate data sets $\vec{x}_i \in \mathbb{R}^p$, $i = 1,\cdots, n$, contain unusual observations that can heavily distort the outcome of a statistical analysis. In particular, it is true for data indexed by a network, which is the main focus of this paper.
The detection of outliers and the development of robust methods are the primary goals of robust statistics. 
Multivariate location and scatter are often at the core of methods used in multivariate analysis and need to be estimated robustly. Arguably, the most useful tool for the detection of multivariate outliers is the squared Mahalanobis distance \cite{maha}. Typically a data point $\vec{x}_i$ is said to be an outlier if its squared Mahalanobis distance $ (\vec{x}_i-\vec{\mu})' \matr{\Sigma}^{-1}(\vec{x}_i-\vec{\mu}) $, where $\vec{\mu}$ and $\vec{\Sigma}$ denote the center and the covariance, exceeds a given threshold. It is well known that $\vec{\mu}$ as well as $\vec{\Sigma}$ need to be estimated robustly \cite{rousseeuw1990unmasking}. Among the most famous estimators is the MCD estimator and its extensions \cite{rousseeuw1985multivariate, rousseeuw1999fast, hubert2010minimum, hubert2018minimum}.

Commonly, observations $\vec{x}_i$ are assumed to be independent of each other. However, additional information often leads to the assumption that there is a dependence between the samples. An obvious case of such a dependence structure is spatially dependent data. Herein, it is often assumed that data points that are close spatially behave similarly, as stated by Tobler's first law \cite{Tobler1970}. Mainly, this comes in the form of an assumption on the pair of points at different locations ($\vec{x}_i,\vec{x}_j)$, $m \neq n$, as is the case in spatial statistics, see \cite{cressie2015statistics,ChilesetAl2012,BaileyEtAl2012}. Recently, methods for the detection of multivariate outliers with spatial dependence have been extended \cite{chen2008detecting, harris2014multivariate, ernst2017comparison}. Of particular interest for this paper is the approach presented in \cite{filzmoser2014identification}, in which outlier detection is based on the squared Mahalanobis distance of pairwise differences $ (\vec{x}_i-\vec{x}_j)' \matr{\Sigma}^{-1}(\vec{x}_i-\vec{x}_j)$. 

The main objective of this paper is to propose a very general framework, as well as robust methods, for the detection of outlying pairs of neighbor points  $(\vec{x}_i,\vec{x}_j)$. An advantage of the proposed method is that the closeness of two points $(\vec{x}_i,\vec{x}_j)$ can be decided beforehand. This need not be  of physical nature. Of course, a typical example would be spatially dependent data where spatially close data should be similar, but the approach presented in this paper also allows for a wider variety of dependence, e.g., personal data in a social network. In addition, we will also allow for external variables to be included, i.e., covariates $z_{i1},\cdots,z_{iq}$ that have an influence on $\vec{x}_n$. 

This paper is structured as follows. In the second section, we introduce the statistical model that we impose on the data matrix $\matr{X}$ inspired by graph signal processing and discuss the general properties of the latter. In the third section, we introduce the notion of, what we dubbed, edgewise outliers. We discuss the  detection of this type of outliers, using robust estimators of location and scale.
The fourth section contains a simulation study and shows the utility of our method in a controlled setting. In the fifth section, we analyze electoral data for the departments of France and comment on the results.

From hereon, matrices are written in bold capital letters, such as $\matr{A}$, and respectively vectors in small bold letters such as $\vec{b}$. Entries will  be written as $a_{ij},$ respectively $b_{i}$. For any matrix $\matr{A} \in  \mathbb{R}^{q \times s}$ we define the column vectorization operator as the operator stacking the columns of the latter into a vector, i.e. $\vecc{(\matr{A})} = (a_{11},\dots,a_{q1},a_{12},\dots,a_{q2},\dots,a_{1s},\dots,a_{qs})'$, where the prime denotes transposition. The inverse of $\vecc$ will be noted as $\vecc^{-1}$.

\section{Probabilistic framework}
In order to introduce the probabilistic model that we consider for the samples $\vec{x}_i$, we recall the definition of the matrix normal distribution and the basics of graph signal theory.

\subsection{Matrix normal distribution} \label{sec:matnorm}
\label{matrnormal}

Let $\matr{X}$ be a real-valued random variable in matrix form of dimension $n \times p$. As in \cite{petersen2008matrix}, we say that $\matr{X}$ follows a matrix  normal distribution if 
\begin{equation} \label{def:matrixNormal}
 \text{vecc}(\matr{X})  \sim \mathcal{N}_{np} (\vecc{(\pmb{\mu})}, \matr{\Sigma}_V \otimes \matr{\Sigma}_G),
\end{equation}
where $\otimes$ denotes the Kronecker product,
$\pmb{\mu} \in \mathbb{R}^{n \times p}$, and $\matr{\Sigma}_V \in \mathbb{R}^{p \times p}$ as well as $\matr{\Sigma}_G \in \mathbb{R}^{n \times n}$ are two positive semi-definite matrices. Alternatively to (\ref{def:matrixNormal}) we also write $\matr{X} \sim \mathcal{N}_{np}(\pmb{\mu},\matr{\Sigma}_G,  \matr{\Sigma}_V)$.
If $\matr{X}$ follows a matrix normal distribution then, thanks to the Kronecker product, the covariance of two entries of $\matr{X}$ can be written as a product of the entries of $\matr{\Sigma}_G$ and $\matr{\Sigma}_V$ as 
\begin{equation} \label{eq:covProd}
    \Cov(X_{ik},X_{jl})=(\matr{\Sigma}_G)_{ij} (\matr{\Sigma}_V)_{kl}.
\end{equation}
Linear transformations of matrix normal distributed variables act accordingly as the following theorem shows. 

\begin{theorem} \label{thm:mattrans}
 If $ \text{vecc}( \matr{X})  \sim \mathcal{N}_{np} (\text{vecc} (\pmb{\mu}), \matr{\Sigma}_V  \otimes \matr{\Sigma}_G)$ then for any two matrices $\pmb{A} \in \mathbb{R}^{m \times n}$ and $\matr{B} \in \mathbb{R}^{p \times q }$ the following holds
\begin{align} \label{eq:transProp}
     \text{vecc}( \matr{A} \matr{X} \matr{B})   \sim \mathcal{N}_{m q} (  \text{vecc}( \matr{A} \pmb{\mu}\matr{B}), (\matr{B}' \matr{\Sigma}_V\matr{B} ) \otimes (\matr{A}  \matr{\Sigma}_G \matr{A}' )).
\end{align}

\end{theorem}

\begin{proof}
A proof is recalled in the appendix. 
\end{proof}

\subsection{Graph signals} \label{sec:graphsig}
Let $(\mathcal{V}, \mathcal{E}) \subset \mathbb{R}^n \times (\mathbb{R}^n \times\mathbb{R}^n )$ be a graph where $\mathcal{V} = \{1,\cdots,n\}$ denotes the set of nodes, and $\mathcal{E}:= \{ (i,j) | w_{ij} \neq 0, i<j \}$  the set of edges with elements $e = (i,j)$ called edges. 
Of central importance in graph theory is the so-called Laplacian matrix $\matr{L} \in \mathbb{R}^{n \times n}$, see \cite{merris1994laplacian}, which can be defined as
\begin{align*}
    \matr{L} = \matr{D} - \matr{W},
\end{align*}
where $\matr{D} $ is a diagonal matrix of the row sums of $\matr{W}$, $d_i := \sum_{j=1}^n w_{ij}$ and $\matr{W}$ is a weight matrix. A weight matrix $\matr{W} \in \mathbb{R}^{n \times n}$ is a symmetric  matrix with $w_{ij} \geq 0$ for $i \neq j$ and $w_{ii} = 0$. The weight matrix can be associated with the graph $(\mathcal{V}, \mathcal{E})$. Graphs are typically used to visualize and capture relations between different nodes $i,j \in \mathcal{V}$, whereas the weights $w_{ij}$ encode the strength of the relation between the former. The larger a weight $w_{ij}$ is, the stronger $i$ and $j$ are related. An important notion is the neighbourhood  $\mathcal{N}(i)$ of a node $i$, i.e. $\mathcal{N}(i) := \{ j  | w_{ij} \neq 0 \}$. It consists of all nodes that are related to $i$.

 The Laplacian matrix encodes important information about the graph structure. Namely, for any vector $\vec{y} \in \mathbb{R}^n,$ the following property holds, see \cite{merris1994laplacian}:
\begin{align} \label{eq:lapprop}
    \vec{y}'\matr{L}\vec{y} = \frac{1}{2} \sum_{i,j = 1}^n (y_i-y_j)^2 w_{ij}.
\end{align}
 Inspecting the right-hand side of (\ref{eq:lapprop}), one can see that if a weight $w_{ij}$ is big, then the difference $y_i-y_j$ contributes more to $\vec{y}'\matr{L}\vec{y}$ than differences corresponding to smaller weights. As mentioned before, the weights $w_{ij}$ encode a presumed relation, and  one can take advantage of (\ref{eq:lapprop}) to define a distribution on $\vec{y} \in \mathbb{R}^n$ such that certain differences are more probable than others. Such model assumptions are regularly made in graph signal processing, see \cite{zhang2015graph,dong2016learning,kalofolias2016learn}, by assuming $\vec{y} \sim \mathcal{N}(\pmb{\mu},\matr{L}^+)$, where $\matr{L}^+$ denotes the Moore-Penrose pseudo-inverse of the matrix $\matr{L}$ and $\pmb{\mu}$ is a vector of $\mathbb{R}^n$. Note that under this model  the density of $\vec{y}$ is given by
\begin{align*}
    p(\vec{y}) \propto \exp\bigg(- \frac{1}{4} \sum_{i,j = 1}^n ((y_i-\mu_i)-(y_j-\mu_j))^2 w_{ij}\bigg).
\end{align*}

\begin{example}[Weights for spatially indexed data] If we can assume that the indices $i = 1,\cdots,n$ of the samples $\vec{x}_i$ refer to spatial positions, say $\vec{s}_i \in \mathbb{R}^2$, then a common choice is to set $w_{ij} = g( \norm{\vec{s}_i-\vec{s}_j}^2)$, where $g:\mathbb{R}_{\geq 0} \rightarrow \mathbb{R}_{\geq 0} $ is a non-increasing function, guaranteeing that the further apart $\vec{s}_i$ and $\vec{s}_j$ are the lower is $w_{ij}.$.Popular choices of $g$ include the Gaussian kernel, respectively the box kernel, leading to 
\begin{align*}
    w_{ij} = \exp\bigg(-\frac{\norm{\vec{s}_i-\vec{s}_j}^2}{2\sigma^2}\bigg),
\end{align*}
respectively
\begin{align*}
    w_{ij} = \vec{1}\{\norm{\vec{s}_i-\vec{s}_j} \leq \sigma \},
\end{align*}
where $\sigma$ is a tuning parameter. The Gaussian kernel will lead to a weight matrix $\matr{W}$ that has no zero entries, due to its smoothness, whereas the box kernel typically leads to sparse weight matrices $\matr{W}$ with entries in $\{0,1\}$.
For completeness, it is also important to mention the  K-nearest neighbors kernel, $K \in \mathbb{N}$, leading to weights
\begin{align*}
    w_{ij} = 
    \begin{cases}
        1 \text{ if } \vec{s}_j \text{ one of the closest $K$ points to } \vec{s}_i \\
        0 \text{ else }.
    \end{cases}
\end{align*}
Typically a post-processing step is applied to symmetrize the weights $w_{ij}$.
\end{example}

For a broader introduction to graph signals, we refer to \cite{stankovic2019graph,stankovic2019graph2,stankovic2020graph3,ortega2018graph}.

\subsection{Network indexed data} \label{subsec:nid}

In light of subsection \ref{sec:matnorm}, the dependence structure  of the samples $\vec{x}_i \in \mathbb{R}^p$ can be determined by the choice of a graph Laplacian $\matr{L}$.  We then assume that the matrix  $\matr{X} \in \mathbb{R}^{n \times p}$, whose rows are the samples $\vec{x}_i$, follows a matrix normal distribution
\begin{align} \label{eq:model}
    \vecc{(\matr{X})} \sim \mathcal{N}_{np} (\vecc{(\pmb{\mu})}, \matr{\Sigma}_V \otimes \matr{L}^+)
\end{align}
with $\matr{\Sigma}_V$ having full rank.
For readability, we write 
 $\vec{x^{\mu}}_i$ for the $i$-th row of $\matr{X}^{\mu}=\matr{X}-\pmb{\mu}$ and $l_{ij}$ for the elements of $\matr{L}^{+}$. Additionally, let $\matr{L}^{+/2}$ denote the square root of $\matr{L}^{+}$.  
Note that the particular case of independence between the vectors $\vec{{x}}_i$, is not a particular case of model \eqref{eq:model} because the identity is not the inverse of any Laplacian matrix.
 
 The following lemma    will be useful in later sections under assumption (\ref{eq:model}). 
\begin{lemma} \label{lem:maha}
    For $\matr{X}$ following  model (\ref{eq:model}), the total Mahalanobis distance $$\md2(\matr{X}) := \vecc(\matr{X}-\pmb{\mu})'(\matr{\Sigma}_V \otimes \matr{L}^+)^{+}\vecc(\matr{X}-\pmb{\mu})$$ can be decomposed as follows:
    \begin{equation} \label{eq:md2L}
    \md2(\matr{X})=\frac{1}{2}\sum_{i,j = 1}^n    (\vec{x^{\mu}}_i- \vec{x^{\mu}}_j)'\matr{\Sigma}_V^{-1} (\vec{x^{\mu}}_i- \vec{x^{\mu}}_j)w_{ij}.
\end{equation}
\end{lemma}

\begin{proof}
See appendix.
\end{proof}

By \eqref{eq:md2L} in the previous lemma  one can see that the density of $\vecc{(\matr{X})} $ is proportional to
\begin{align*}
    \exp\bigg( -\frac{1}{4}  \sum_{i,j = 1}^n    (\vec{x^{\mu}}_i- \vec{x^{\mu}}_j)'\matr{\Sigma}_V^{-1} (\vec{x^{\mu}}_i- \vec{x^{\mu}}_j)w_{ij} \bigg),
\end{align*}
which provides an insight into the effect of the magnitude of the weights $w_{ij}$. Similar to the discussion in subsection \ref{sec:graphsig}, differences $\vec{x^{\mu}}_i- \vec{x^{\mu}}_j$ over edges $(i,j)$ corresponding to higher weights are the most influential. 
As the density of $\vecc{(\matr{X})} $ depends only on $\Delta_{ij} := (\vec{x^{\mu}}_i- \vec{x^{\mu}}_j)'\matr{\Sigma}_V^{-1} (\vec{x^{\mu}}_i- \vec{x^{\mu}}_j)w_{ij}, $ we derive in the following lemma the distribution of the $\Delta_{ij}$.

\begin{lemma} \label{lem:deltadistr}
For $\matr{X}$ following  model (\ref{eq:model}),
 we  have 
\begin{align}
    \Delta_{ij} \sim w_{ij}(l_{ii}+l_{jj}-2l_{ij})\chi^2(p) \label{eq:deltradistr}
\end{align}
for any $i,j \in \{1,\dots,n \}$.


\end{lemma}

\begin{proof}
See appendix.
\end{proof}

Lemma \ref{lem:deltadistr}  allows for a decision rule for the detection of edgewise outliers, where typically $\tau$ is chosen as $\chi^2_{p,0.975}$.

\paragraph{Edge outlier detection rule} \label{def:outl}
We say that an edge $(i,j) \in \mathcal{E}$ is an outlier if 
$\frac{\Delta_{ij}}{ w_{ij}(l_{ii}+l_{jj}-2l_{ij})} > \tau $.\\

This decision rule assumes that we know the model parameters $\pmb{\mu}$ and $\matr{\Sigma}_V.$ The following section treats their maximum likelihood estimation and introduces a robust alternative.

\begin{remark}
Note that for any node $i$ we also have ${\vec{x^{\mu}}_i}'\matr{\Sigma}_V^{-1} \vec{x^{\mu}}_i \sim l_{ii}\chi^2(p) $. With this, we could also define node outliers as a node $i$ for which $\frac{{\vec{x^{\mu}}_i}'\matr{\Sigma}_V^{-1} \vec{x^{\mu}}_i }{l_{ii}} > \tau$ holds. 
\end{remark}

\section{Estimation}
\subsection{Maximum likelihood estimation}
In the following, we allow for the mean function $\pmb{\mu}$ to be parametrized by some $\pmb{\theta} \in \mathbb{R}^{q \times \Tilde{p}}$ where $\Tilde{p} \leq p$, which is usually taken equal to the data dimension $p$. The following theorem derives the  maximum likelihood equations for $(\pmb{\theta},\matr{\Sigma}_V)$, which is a special case of the multivariate normal regression model, see \cite{seber2009multivariate}. 

\begin{theorem} \label{thm::likelihoodestimators}
    Assume that $\pmb{\mu}$ is parametrized by $\pmb{\theta} \in \mathbb{R}^{q \times \Tilde{p}}$. Then the maximum likelihood estimators for $\matr{\Sigma}_V$ and $\vec{\theta}$ with $\vecc{(\matr{X})} \sim \mathcal{N}_{np} ( \vecc{(\pmb{\mu}(\vec{\theta}))}, \matr{\Sigma}_V \otimes \matr{L}^+)$ satisfy the following equations
\begin{align*}
    & \sum_{i=1}^n \sum_{k = 1}^p (\matr{L}(\pmb{\mu}(\vec{\theta})-\matr{X})\matr{\Sigma}_V^{-1})_{ik}\frac{\partial {\mu}(\vec{\theta})_{ik}}{\partial {\theta}_{ml}} = 0 \text{ for }  \,\, \text{ and } \,\, {\matr{\Sigma}}_V= \frac{1}{n} (\matr{X}-\pmb{\mu}(\vec{\theta}))'\matr{L}(\matr{X}-\pmb{\mu}(\vec{\theta})).
\end{align*}
If $\pmb{\mu}(\vec{\theta})$ is modeled as $\pmb{\mu}(\pmb{\theta}) = \matr{Z}\pmb{\theta}$ for a fixed $\matr{Z} \in \mathbb{R}^{n \times q }$ and $\pmb{\theta} \in \mathbb{R}^{q \times p}$ , which implies that $\tilde{p}=p$, then these equations become 
\begin{align*}
     &  \matr{Z}'\matr{L}  \matr{Z} \pmb{\theta} =  \matr{Z}'\matr{L} \matr{X}\,\, \text{ and } \,\, {\matr{\Sigma}}_V= \frac{1}{n} (\matr{X}-\matr{Z} \pmb{\theta})'\matr{L}(\matr{X}-\matr{Z} \pmb{\theta}).
\end{align*}

For $\matr{Z} = \vec{1}_n$ and $\pmb{\theta} \in \mathbb{R}^{1 \times p}$ the first equation holds for any $\pmb{\theta}$ and the estimator for $\matr{\Sigma}_V$ is
\begin{align*}
    {\matr{\Sigma}}_V= \frac{1}{n} \matr{X}'\matr{L}\matr{X}.
\end{align*}

\end{theorem}

\begin{proof}
See appendix.
\end{proof}

If $\pmb{\mu}{(\pmb{\theta})}$ is of the form $\pmb{\mu}{(\pmb{\theta})} = \matr{Z}\pmb{\theta}$, we know by Theorem \ref{thm::likelihoodestimators} that the estimators can be written as 
\begin{align} 
    &\pmb{\hat \theta} =  (\matr{Z}'\matr{L}  \matr{Z})^{-1}\matr{Z}'\matr{L} \matr{X}\label{eq:estMu} \\
    &{\matr{\hat \Sigma}}_V= \frac{1}{n} (\matr{X}-\matr{Z} \pmb{\hat \theta})'\matr{L}(\matr{X}-\matr{Z} \pmb{\hat \theta}) \label{eq:estCovV}.
\end{align}

Using another property of the Laplacian matrix $\matr{L}$, we can rewrite equations (\ref{eq:estMu}) and (\ref{eq:estCovV}). It is well known, see \cite{grady2010discrete}, that $\matr{L}$ can also be written as $\matr{L}:= \matr{M}' \matr{M}$, where $\matr{M} \in \mathbb{R}^{|\mathcal{E}| \times n}$ is defined entrywise for each edge $e=(i,j) \in \mathcal{E}$ and $l \in \{1,\dots,n\}$ as
\begin{align*}
    \matr{M}_{(i,j),l} := \begin{cases}
     \sqrt{w_{ij}}& \text{ if l = i } \\
     -\sqrt{w_{ij}}& \text{ if l = j } \\
     0 & \text{ else }
    \end{cases}.
\end{align*}

For any two matrices $\matr{A}$ and $\matr{B}$ of appropriate dimensions we can therefore rewrite a matrix product of the form $\matr{A}'\matr{L}\matr{B}$, as $\matr{A}'\matr{L}\matr{B} = (\matr{M} \matr{A})'(\matr{M} \matr{B}) = \sum_{i<j} (\vec{a}_{i,:}-\vec{a}_{j,:})(\vec{b}_{i,:}-\vec{b}_{j,:})'w_{ij}$, where $\vec{a}_{i,:}$ denotes the i-th row of $\matr{A}$; and similarly for the other subscripts.  The terms $\vec{a}_{i,:}-\vec{a}_{j,:}$ and $\vec{b}_{i,:}-\vec{b}_{j,:}$ are matrix terms differences from one node $i$ to another node $j$. 
Applying this to the matrix products in (\ref{eq:estMu}) and (\ref{eq:estCovV}) we have
\begin{align}
& \pmb{\hat \theta} =  \Big( \sum_{i<j} (\vec{z}_{i,:}-\vec{z}_{j,:})(\vec{z}_{i,:}-\vec{z}_{j,:})'w_{ij} \Big)^{-1} \Big(\sum_{i<j} (\vec{z}_{i,:}-\vec{z}_{j,:})(\vec{x}_{i}-\vec{x}_{j})'w_{ij} \Big)  \label{eq:estMuEdge} \\
 &  {\matr{\hat \Sigma}}_V = \frac{1}{n} \sum_{i<j}(\vec{x^{\mu}}_i- \vec{x^{\mu}}_j)(\vec{x^{\mu}}_i- \vec{x^{\mu}}_j)'w_{ij} \label{eq:estCovVEdge}.
\end{align}
From this, it is obvious that abnormal edgewise differences $\vec{z}_{i,:}-\vec{z}_{j,:}$ or $\vec{x^{\mu}}_i- \vec{x^{\mu}}_j$ can lead to distorted estimations in (\ref{eq:estMu}) and (\ref{eq:estCovV}). 
These estimators need therefore to be  robustified. So far we assumed model (\ref{eq:model}) and deduced that the squared Mahalanobis distance can be written, see equation (\ref{eq:md2L}), in terms of $(\vec{x^{\mu}}_i- \vec{x^{\mu}}_j)'\matr{\Sigma}_V^{-1} (\vec{x^{\mu}}_i- \vec{x^{\mu}}_j)w_{ij}$. As the  maximum likelihood estimators (\ref{eq:estMu}) and (\ref{eq:estCovV}) are derived from minimizing the negative log-likelihood
\begin{align} \label{opt:edgewiseopt} 
     \sum_{i < j} (\vec{x^{\mu}}_i- \vec{x^{\mu}}_j)'\matr{\Sigma}_V^{-1} (\vec{x^{\mu}}_i- \vec{x^{\mu}}_j)w_{ij}  + n \log(|{\matr{\Sigma}}_V|)
 \end{align}
with respect to $(\matr{\vec{\theta}},\matr{\Sigma}_\matr{V})$, where the edgewise differences $\vec{x^{\mu}}_i- \vec{x^{\mu}}_j$
appear again, it seems natural to take a trimmed approach to robustify the estimators.

\subsection{Robust estimation with edgewise MCD} \label{subsec:edgemcd}


Trimmed estimators have been frequently used \cite{rousseeuw1984least,rousseeuw1985multivariate,bednarski1993trimmed,hadi1997maximum} in the case of outliers.
Applying the idea of trimming to (\ref{opt:edgewiseopt}) we propose to solve the following problem to find robust estimates $\vec{\theta}$ and ${\matr{\Sigma}}_{\matr{V}}$:
 \begin{align} \label{edgeMCD}
     \min_{\matr{\vec{\theta}},{\matr{\Sigma}}_V}\sum_{k = 1}^{h}(\vec{x^{\mu}}_{\pi(k)_1}-\vec{x^{\mu}}_{\pi(k)_2})'{\matr{\Sigma}}_V^{-1}(\vec{x^{\mu}}_{\pi(k)_1}-\vec{x^{\mu}}_{\pi(k)_2})w_{\pi(k)}  + (n h |\mathcal{E}|^{-1}) \log(|{\matr{\Sigma}}_V|).
 \end{align}
 where, $h\in [\frac{|\mathcal{E}|+p+1}{2},|\mathcal{E}|]$ denotes a cutoff to be specified, and $\pi$ is the function mapping an element $\{1,\dots,|\mathcal{E}|\}$ to an edge $(i,j) \in \mathcal{E}$, corresponding to the ordering, from lowest to highest, of $\Delta_{ij}$, i.e. $\Delta_{\pi(1)} \leq \Delta_{\pi(2)}\leq \Delta_{\pi(3)} \leq \dots.$ We recall that $\Delta_{ij}:=(\vec{x^{\mu}}_i- \vec{x^{\mu}}_j)'{\matr{\Sigma}}_V^{-1} (\vec{x^{\mu}}_i- \vec{x^{\mu}}_j)w_{ij}$ and we denote $\pi(k) = (\pi(k)_1,\pi(k)_2)$. 
 Remark that for $h=|\mathcal{E}|$, solving \eqref{edgeMCD} is equivalent to solving  \eqref{opt:edgewiseopt} and leads to the estimators (\ref{eq:estMu}) and (\ref{eq:estCovV}). For $h < |\mathcal{E}|$, an edge $(i,j)$ with comparatively  high $\Delta_{ij}$ has no influence as it does not appear in (\ref{edgeMCD}) and will not distort the estimators.

Given initial estimates $\matr{\vec{\theta}}^{0}$ and ${\matr{\Sigma}}_V^{0},$ we set $\vec{x^{\mu^0}}_i = \vec{x}_i - \matr{\vec{\theta}}^{0}\vec{z}_i$ for all $i=1,\dots,n$, and compute updated estimates by the following steps. This algorithm which employs an edgewise version of the C-step of the popular MCD algorithm, see \cite{rousseeuw1985multivariate,rousseeuw1999fast}, iterates between finding the currently most probable samples and updating the parameters:
\begin{itemize}
    \item[0)] Set t=0.
    \item[1)] Order for each edge $(i,j) \in \mathcal{E}$ the quantities $\Delta_{ij}^t:=(\vec{x^{\mu^{t}}}_i- \vec{x^{\mu^{t}}}_j)'({\matr{\Sigma}}_V^{t})^{-1} (\vec{x^{\mu^{t}}}_i- \vec{x^{\mu^{t}}}_j)w_{ij}$ from lowest to largest and denote $\pi^t$ the corresponding order function.
    \item[2)] Update $\matr{\vec{\theta}}^{t}$ by:
    \begin{align*}
       &\matr{\Gamma}_{\matr{Z}'\matr{Z}}:=\sum_{k = 1}^{h}(\vec{z}_{\pi^t(k)_1}-\vec{z}_{\pi^t(k)_2})(\vec{z}_{\pi^t(k)_1}-\vec{z}_{\pi^t(k)_2})'w_{\pi^t(k)} \\
       &\matr{\Gamma}_{\matr{Z}'\matr{X}}:=\sum_{k = 1}^{h}(\vec{z}_{\pi^t(k)_1}-\vec{z}_{\pi^t(k)_2})(\vec{x}_{\pi^t(k)_1}-\vec{x}_{\pi^t(k)_2})'w_{\pi^t(k)} \\
       &\matr{\vec{\theta}}^{t+1} \leftarrow \matr{\Gamma}_{\matr{Z}'\matr{Z}}^{-1}\matr{\Gamma}_{\matr{Z}'\matr{X}}.
    \end{align*}
  \item[3)]  Update $\vec{x^{\mu^{t}}}$: $\vec{x^{\mu^{t+1}}}_i \leftarrow \vec{x}_i - \matr{\vec{\theta}}^{t+1}\vec{z}_i$ for all $i=1,\dots,n$.  
  \item[4)] Update ${\matr{\Sigma}}_V^{t}$: 
    \begin{align*}
        {\matr{\Sigma}}_V^{t+1} \leftarrow \frac{1}{n h |\mathcal{E}|^{-1}}\sum_{k = 1}^{h}(\vec{x^{\mu^{t+1}}}_{\pi^t(k)_1}-\vec{x^{\mu^{t+1}}}_{\pi^t(k)_2})(\vec{x^{\mu^{t+1}}}_{\pi^t(k)_1}-\vec{x^{\mu^{t+1}}}_{\pi^t(k)_2})'w_{\pi^t(k)}
    \end{align*}
    \item[5)] Increment $t$ and go back to step 1) until convergence. 
\end{itemize}

As in \cite{rousseeuw1999fast}, to prove that the suggested algorithm decreases the objective in every step, we will show the following chain of inequalities
\begin{align*}
     \sum_{k = 1}^{h}(\vec{x^{\mu^t}}_{\pi^t(k)_1}-&\vec{x^{\mu^t}}_{\pi^t(k)_2})'({\matr{\Sigma}}_V^{t})^{-1}(\vec{x^{\mu^t}}_{\pi^t(k)_1}-\vec{x^{\mu^t}}_{\pi^t(k)_2})'w_{\pi^t(k)} + (n h |\mathcal{E}|^{-1}) \log(|{\matr{\Sigma}}_V^t|) \\
     & \geq  \sum_{k = 1}^{h}(\vec{x^{\mu^{t+1}}}_{\pi^t(k)_1}-\vec{x^{\mu^{t+1}}}_{\pi^t(k)_2})'({\matr{\Sigma}}_V^{t+1})^{-1}(\vec{x^{\mu^{t+1}}}_{\pi^t(k)_1}-\vec{x^{\mu^{t+1}}}_{\pi^t(k)_2})'w_{\pi^t(k)}  +  (n h |\mathcal{E}|^{-1}) \log(|{\matr{\Sigma}}_V^{t+1}|) \\
     & \geq \sum_{k = 1}^{h}(\vec{x^{\mu^{t+1}}}_{\pi^{t+1}(k)_1}-\vec{x^{\mu^{t+1}}}_{\pi^{t+1}(k)_2})'({\matr{\Sigma}}_V^{t+1})^{-1}(\vec{x^{\mu^{t+1}}}_{\pi^{t+1}(k)_1}-\vec{x^{\mu^{t+1}}}_{\pi^{t+1}(k)_2})'w_{\pi^{t+1}(k)} +  (n h |\mathcal{E}|^{-1}) \log(|{\matr{\Sigma}}_V^{t+1}|)
\end{align*}
where $\pi^t$ is the ordering function from Step 1) at the $t$-th iteration and $\pi^{t+1}$ is the ordering function from Step 1) at the ($t$+1)-th iteration. The second inequality is trivial as $\pi^{t+1}$ maps exactly to those edges such that $\Delta_{ij}^{t+1}$ is ordered from lowest to highest. For the first inequality, note that $\matr{\vec{\theta}}^{t+1}$ and ${\matr{\Sigma}}_V^{t+1}$ are the minimizers of problem (\ref{edgeMCD}), with $\pi = \pi^t$ and respectively $\matr{\vec{\theta}}^{t+1}$ fixed. As the number of possible permutations $\pi$ is finite we can conclude that this algorithm converges. Typically we experienced convergence in less than ten cycles. 
We then perform one final reweighting step similar to the MCD algorithm \cite{lopuhaa1991breakdown,lopuhaa1999asymptotics}. For this, we make use of the distribution of $\Delta_{ij}$ as in (\ref{eq:deltradistr}) as follows. We perform the updates of step 2) to 4) for $\matr{\vec{\theta}}^{T}$ and ${\matr{\Sigma}}_V^{T}$ with $\pi$ mapping into the set of all edges $(i,j) \in \mathcal{E}$ with $\frac{\Delta_{ij}^T}{ w_{ij}(l_{ii}+l_{jj}-2l_{ij})} \leq \chi^2_{p,0.975} $, where $\chi^2_{p,0.975}$ is the $0.975$ quantile of a chi-square distributed variable with p degrees. Finally, we rescale the covariance estimate ${\matr{\Sigma}}_V^{T}$ by replacing it with ${\matr{\Sigma}}_V^{T+1} = c {\matr{\Sigma}}_V^{T}$, where $c$ is a constant such that $\med_{(i,j) \in \mathcal{E}}(\frac{\Delta_{ij}^T}{ w_{ij}(l_{ii}+l_{jj}-2l_{ij})}) = \chi^2_{p,0.5}$. 

Good initial starting parameters are essential. We consider four different initial estimates, similar to the deterministic MCD algorithm, see \cite{hubert2012deterministic}. We compute for each edge $(i,j) \in \mathcal{E}$ the weighted edgewise differences $(\vec{z}_i - \vec{z}_j)\sqrt{w_{ij}}$ and $(\vec{x}_i - \vec{x}_j)\sqrt{w_{ij}}$ and put these row-wise into the data matrix $\matr{Z}_{\mathcal{E}} \in \mathbb{R}^{|\mathcal{E}| \times q}$ and $\matr{X}_{\mathcal{E}} \in \mathbb{R}^{|\mathcal{E}| \times p}$. Note that with this, the solutions (\ref{eq:estMuEdge})-(\ref{eq:estCovVEdge}) can be rewritten in matrix notation
\begin{align} 
    &\pmb{\theta} =  (\matr{Z}_{\mathcal{E}}'\matr{Z}_{\mathcal{E}})^{-1}\matr{Z}_{\mathcal{E}}' \matr{X}_{\mathcal{E}}  \label{eq:edgewiseTheta} \\
    &{\matr{\Sigma}}_V= \frac{1}{n} (\matr{X}_{\mathcal{E}}-\matr{Z}_{\mathcal{E}} \pmb{\theta})'(\matr{X}_{\mathcal{E}}-\matr{Z}_{\mathcal{E}} \pmb{\theta}). \label{eq:edgewiseCov} 
\end{align}
As the matrix products involving $\matr{Z}_{\mathcal{E}}$, $\matr{X}_{\mathcal{E}}$ and $\matr{X}_{\mathcal{E}}-\matr{Z}_{\mathcal{E}} \pmb{\theta}$ can be thought of in terms of the population versions of covariance estimates, we compute robust starting estimates $\matr{\vec{\theta}}^{0}$ and ${\matr{\Sigma}}_V^{0}$ similar to the steps described in the deterministic MCD \cite{hubert2012deterministic} by transforming the data. 
We describe the method in terms of two general data matrices $\matr{R}$ and $\matr{T}$, which take the role of $\matr{Z}_{\mathcal{E}}$, $\matr{X}_{\mathcal{E}}$ and $\matr{X}_{\mathcal{E}}-\matr{Z}_{\mathcal{E}}\pmb{\theta}$. First, each column of $\matr{R}$ and $\matr{T}$ is scaled with a robust scale estimate to get  $\matr{R}^s$ and $\matr{T}^s$. In this paper, we use the robust Qn-scale estimator, see \cite{rousseeuw1993alternatives}.
Three estimators can be computed by column-wise transformations. Denote $\vec{u}$ any column of $\matr{R}^s$ resp. $\matr{T}^s$, and $\widetilde{\matr{R}}^s$ resp. $\widetilde{\matr{T}}^s$ the resulting matrices by applying one of the following functions to each column: 
\begin{itemize}
    \item $\phi_1(\vec{u}):=\tanh{(\vec{u})}$
    \item $\phi_2(\vec{u}):= rank{(\vec{u})}$, where $rank$ denotes the ranking operation of the entries of $\vec{u}$,  
    \item $\phi_3(\vec{u}):=\phi\bigg(\frac{(\phi_2(\vec{u})-\frac{1}{3})}{|\mathcal{E}|+\frac{1}{3}}\bigg)$ where $\phi$ is the normal cumulative distribution function.
\end{itemize}
Then in each case, an estimator of the correlation is given by $\matr{S} = \corr(\widetilde{\matr{R}}^s,\widetilde{\matr{T}}^s)$. A fourth correlation estimator can be obtained by a row-wise transformation: applying to each row $\vec{u}$ of $\matr{R}^s$ resp. $\matr{T}^s$ the function $\psi_4(\vec{u})= \frac{\vec{u}}{\norm{\vec{u}}}_2$ to obtain $\widetilde{\matr{R}}^s$ resp. $\widetilde{\matr{T}}^s$, the spatial sign estimator is then given as $\matr{S} = \frac{1}{|\mathcal{E}|} \widetilde{\matr{R}}^{s'}\widetilde{\matr{T}}^s$, see also \cite{durre2015spatial}.
Again, similar to the deterministic MCD, we adjust the singular values of $\matr{S}$. After computing the SVD of $\matr{S}$: $\matr{S} = \matr{U} \matr{\Sigma} \matr{V}',$ we apply the following steps.
\begin{itemize}
    \item Compute the projections  $ {\matr{R}}^u:= {\matr{R}}^s\matr{U}$ and ${\matr{T}}^u :={\matr{T}}^s\matr{V}$.
    \item Robustly estimate the scales of ${\matr{R}}^u$ and ${\matr{T}}^u$ with the Qn-scale estimator. Denote these estimates $\vec{\sigma}_{{\matr{R}}^u}$ and $\vec{\sigma}_{{\matr{T}}^u}$.
    \item Replace $\matr{S}$ by $ \matr{U} \diag{(\vec{\sigma}_{{\matr{R}}^u})} \diag{(\vec{\sigma}_{{\matr{T}}^u})} \matr{V}' $, where $\diag$ denotes a diagonal matrix of corresponding entries. 
\end{itemize}
Finally, to obtain a covariance estimate we transform $\matr{S}$ back by multiplying each column, respectively row, with the originally estimated scales. Plugging in $\matr{Z}_{\mathcal{E}}$, $\matr{X}_{\mathcal{E}}$ or $\matr{X}_{\mathcal{E}}-\matr{Z}_{\mathcal{E}}\pmb{\theta}$ for 
$\matr{R}$ or $\matr{T}$, computing covariance estimates for the latter and using (\ref{eq:edgewiseTheta}) and (\ref{eq:edgewiseCov}) leads to initial estimators of $\matr{\vec{\theta}}^{0}$ and ${\matr{\Sigma}}_V^{0}$.


\section{Simulation Study} \label{sec:sim}

We perform a simulation study to test the utility of the proposed method. To generate data from the model (\ref{def:matrixNormal}), we need to first select a covariance matrix $\matr{\Sigma}_{V}$ and a graph Laplacian matrix $\matr{L}$. We start by fixing a dimension $p \in \{3,10\}$ and a number of nodes $n \in \{ 50,100,200,300 \}$. We define the covariance matrix $\matr{\Sigma}_{V}$  by drawing independent entries of a matrix  $\matr{H}\in \mathbb{R}^{p\times p}$ from a standard Gaussian and computing the eigenvectors $\matr{U}$ of $\matr{H}'\matr{H}$. Sampling eigenvalues $\vec{\sigma} \in \mathbb{R}^p$ from a uniform distribution, $\sigma_k \sim \mathcal{U}[1,50]$, we obtain a covariance matrix by $\matr{\Sigma}_{V} = \matr{U} \diag(\vec{\sigma}) \matr{U}'$.
To obtain the Laplacian matrices $\matr{L}$, we  select three graphs obtained by simulating from three different types of graphs covering a wide range of models of connectivity: 
\begin{itemize}
    \item {Knn graph:} we simulate $n$-times a coordinate vector $(x_i,y_i) \in \mathbb{R}^2$, $i=1,\dots,n$, with $x_i,y_i \sim \mathcal{U}[0,1]$ and compute the five nearest neighbors for each node. The edge set $\mathcal{E}$ is then defined by connecting the closest five  neighbors, see \cite{marchette2005random}, using the R package \cite{cccd}.
    \item {Erdos-Renyi graph:} we generate $n$ nodes and connect at random any two nodes with a probability of 0.05. This Erdos-Renyi graph \cite{gilbert1959random} is obtained using the R package \cite{igraph}. The edge set $\mathcal{E}$ is then given by the graph structure.  
    \item {Scalefree graph:} we generate a graph from the Barabási-Abert model with parameters $(m0 = 1, m = 2)$, see \cite{barabasi1999emergence}. The edge set $\mathcal{E}$ is then given by the graph structure.  
\end{itemize}

Next we generate corresponding weights $w_{ij}\sim \mathcal{U}[0,1]$ for $(i,j) \in \mathcal{E}$ and denote the resulting matrix by $\matr{W}$. Then we set $\matr{L} = \diag({\matr{W}} \vec{1})-{\matr{W}} $.
To get a data matrix $\matr{X}$ following the distribution (\ref{eq:model}) with given $\matr{\Sigma}_{V}$ and $\matr{L}$, we do the following steps:
\begin{enumerate}
    \item Calculate the square root of $\matr{\Sigma}_V = \matr{\Sigma}_V ^{\frac{1}{2}} \matr{\Sigma}_V^{\frac{1}{2}}$ and the square root of the generalized inverse $ \matr{L}^+ = \matr{L}^{+/2} \matr{L} ^{+/2} $.
    \item Draw $\text{vecc} (\matr{Y}) \sim  \mathcal{N}_{np} (\matr{0}_{np}, \matr{I}_p \otimes \matr{I}_n)$.
    \item Draw independently entries $z_{im} \sim \mathcal{U}[-1,1]$, $i=1,\cdots,n$ and $m=1,\cdots,q$, with $q = 7$, to obtain a covariate matrix $\matr{Z} \in \mathbb{R}^{n \times q}$. Additionally, draw a coefficient matrix $\pmb{\theta} \in \mathbb{R}^{ q \times p} $ with i.i.d entries $\theta_{ml} \sim \mathcal{N}(0,1)$, with $m=1,\cdots,q$ and $l=1,\cdots,p$.
    \item Finally, a matrix $\matr{X}$ that follows the distribution $\text{vecc} (\matr{X}) \sim \mathcal{N}_{np} (\text{vecc}(\pmb{\mu}),\matr{L}^+  \otimes \matr{\Sigma}_V)$ is obtained by setting $\matr{X} =\pmb{\mu}  + \matr{L} ^{+/2}\matr{Y} \matr{\Sigma}_V ^{\frac{1}{2}}$ with $\pmb{\mu}:=\matr{Z} \pmb{\theta} $.
\end{enumerate}

Finally, we also corrupt the data by creating edgewise outliers. First, we fix a percentage of edges to be corrupted $\zeta \in  \{ 0.05,0.1,0.2,0.3 \} \times 100 \%$. 
We corrupt the data matrices $\matr{X}$ and $\matr{Z}$ in the following way. We denote by $\vec{r}$ the eigenvector associated to the largest eigenvalue of ${\matr{\Sigma}}_V$. We order the entries of $\matr{X} \vec{r} \in \mathbb{R}^n$ from lowest to highest and denote $\gamma$ its order function. Then we swap the rows $\vec{x}_{\gamma(1)},\vec{x}_{\gamma(2)},\dots,\vec{x}_{\gamma(k)}$ with the rows $\vec{x}_{\gamma(n)},\vec{x}_{\gamma(n-1)},\dots,\vec{x}_{\gamma(n-k+1)}$ for a $k$ such that at most $\zeta |\mathcal{E}|$  edges are affected. This is similar to the corruption setting suggested in \cite{harris2014multivariate}. Denote by ${\cal V}_{\text{corr}}$ the set of nodes that have been corrupted this way. 
Next, we also corrupt the covariates data matrix $\matr{Z}$. To do so we replace each $\vec{z}_i$ with $i \in  {\cal V}_{\text{corr}}$ by a multivariate point with entries sampled from $\mathcal{U}[-10,10]$. 
To compare the performance of the estimated parameters $(\hat{\matr{\vec{\theta}}},{\hat{\matr{\Sigma}}}_V)$  to the true ones $(\matr{\vec{\theta}},{\matr{\Sigma}}_V)$, we use three different error scores, the F-score (Fsc), the Kullback–Leibler divergence (KL) and the relative distance (RD). Define the set $\mathcal{E}_{outl}$ of edges that are edge outliers given the parameters $(\matr{\vec{\theta}},{\matr{\Sigma}}_V)$, i.e. 
$$\mathcal{E}_{outl} := \Big\{ (i,j) \in \mathcal{E}: \frac{(\vec{x^{\mu}}_i- \vec{x^{\mu}}_j)'\matr{\Sigma}_V^{-1} (\vec{x^{\mu}}_i- \vec{x^{\mu}}_j)}{ w_{ij}(l_{ii}+l_{jj}-2l_{ij})} > \chi^2(p,0.95) \Big\},$$
and in the same manner define $\widehat{\mathcal{E}}_{outl}$ for estimated $(\widehat{\matr{\vec{\theta}}},{\widehat{\matr{\Sigma}}}_V)$. Then the scores Fsc, KL and RD are defined as:
\begin{itemize}
    \item $ \text{Fsc}:= 2 \frac{ \text{Pr} \cdot \text{Rec}}{\text{Pr} + \text{Rec}}$, where Pr is the precision $\text{Pr}:=\frac{|\mathcal{E}_{outl} \cap \widehat{\mathcal{E}}_{outl}|}{|\mathcal{E}_{outl}|} $ and Rec the recall $\text{Rec}:=\frac{|\mathcal{E}_{outl} \cap \widehat{\mathcal{E}}_{outl}|}{|\widehat{\mathcal{E}}_{outl}|}$.
\item 
$\text{KL}(\hat{\matr{\vec{\theta}}},\matr{\vec{\theta}},\hat{\matr{\Sigma}}_V,{\matr{\Sigma}}_V):=\frac{1}{2} \bigg( \Tr\big( {\matr{\Sigma}}_V^{-1} \hat{\matr{\Sigma}}_V \big) -p + \log\bigg( \frac{|{\hat{\matr{\Sigma}}}_V|}{|{{\matr{\Sigma}}}_V|} \bigg) + \frac{1}{q}\Tr{\big( \big( \matr{Z}\matr{\vec{\theta}} - \matr{Z}\hat{\matr{\vec{\theta}}} \big){\matr{\Sigma}}_V^{-1} \big( \matr{Z}\matr{\vec{\theta}} - \matr{Z}\hat{\matr{\vec{\theta}}} \big)'}\big) \bigg)$
\item 
$\text{RD}(\hat{\matr{\vec{\theta}}},\matr{\vec{\theta}}) := \frac{\norm{\hat{\matr{\vec{\theta}}} - {\matr{\vec{\theta}}}}_F}{\norm{\matr{\vec{\theta}}}_F}$.
\end{itemize}

We compare the error scores Fsc, KL and RD for the  method edgemcd proposed in subsection \ref{subsec:edgemcd} to the  deterministic MCD method \cite{hubert2012deterministic} on $\matr{X}_{\mathcal{E}}-\matr{Z}_{\mathcal{E}} \hat{\pmb{\theta}}$, where we robustly estimate $\hat{\pmb{\theta}}$ beforehand by LTS regression \cite{rousseeuw1984least}  with the R package \cite{robustbase}, and to the standard std estimators std (\ref{eq:estMu})-(\ref{eq:estCovV}).
Figure \ref{fig::kl}, \ref{fig::theta} and \ref{fig::fs} display the performance measures KL, RD, and Fsc for these three methods depending on the graph type, the percentage of corruption, for a growing number of nodes N and different dimensions $p$. We can see in Figure \ref{fig::kl} that  for all three graph structures and for no corruption (0\%) the proposed edgemcd method does not perform considerably worse than the standard std estimates. For growing corruption rate and growing number of nodes, the Kullback-Leibler divergence grows considerably for the mcd and std methods, whereas our edgemcd method still improves with growing $N$, i.e. the estimates for $\matr{\Sigma}_V$ and $\pmb{\theta}$ improve with a growing number of nodes. This can also be seen in Figure \ref{fig::theta} which displays the relative distance for the estimated coefficients $\vec{\theta}$. 
The standard std estimates continue the give bad results. Even though the mcd estimates improve with a growing number of nodes, our edgemcd method shows smaller errors and a quicker improvement. Finally, Figure \ref{fig::fs} shows the performance in terms of F-score. Clearly, the F-score for the standard std estimates is becoming worse and worse with a growing percentage of corruption even though with a growing number of nodes it slightly improves. The mcd method performs considerably well in comparison to std for the knn and scalefree graphs structure with a growing number of nodes but becomes quickly worse with $N$ for Erdos-Renyi graphs. Our method outperforms std and mcd in these settings. One might wonder why the F-scores become better with the growing percentage of outliers. The reason for this has to do with the outlier generating process as described at the beginning of this section. As we allow a higher percentage of edges to be outliers we also allow for a higher percentage of corrupted nodes. If both nodes of an edge $(i,j)$ are corrupted, then this has an even  larger effect on $\Delta_{ij}$ and makes these edges easier to detect if the estimates for $\matr{\Sigma}_V$ and $\pmb{\theta}$ are reasonably good.

\begin{figure}[h]
    \centering
    \includegraphics[scale = 0.28]{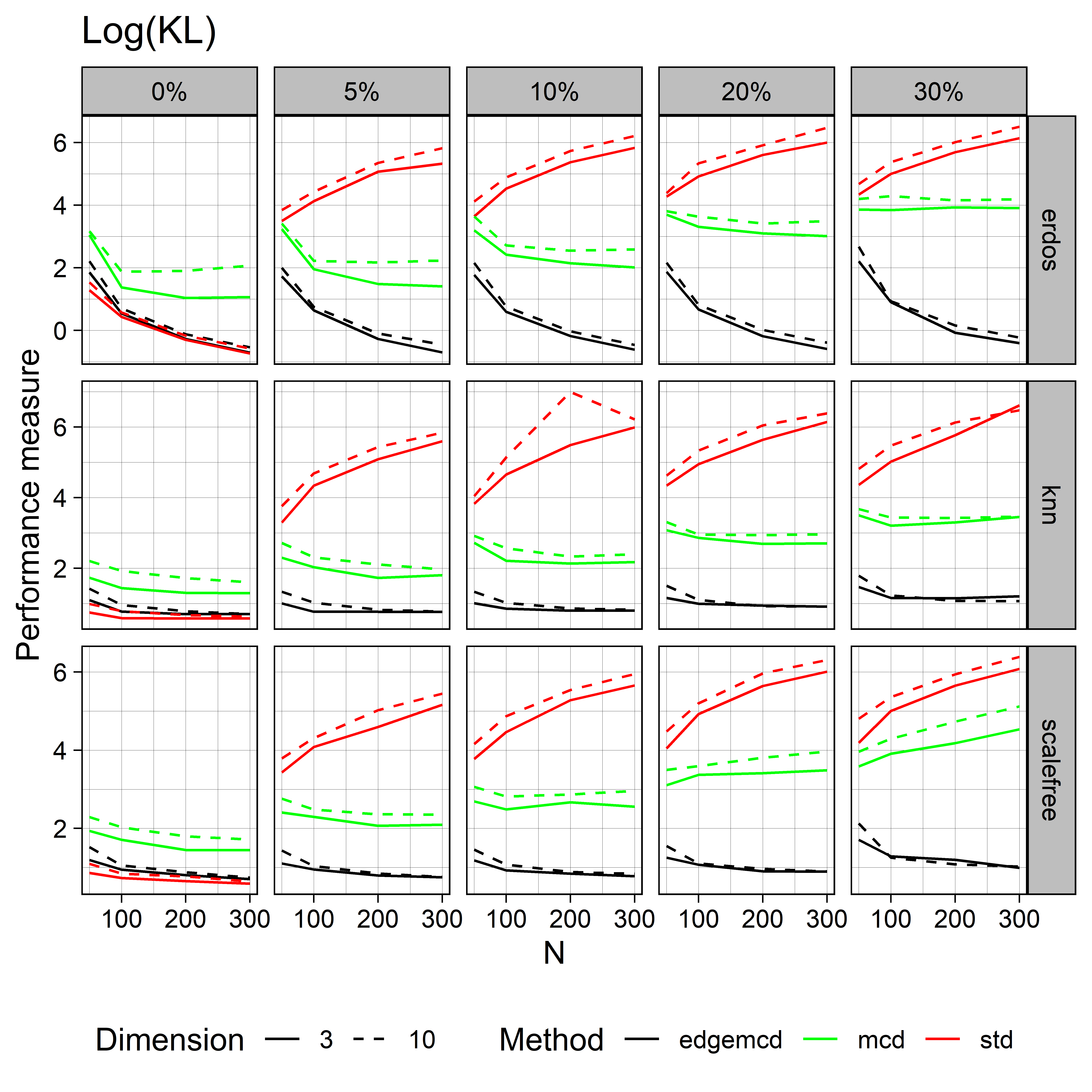}
    \caption{Log transformed Kullback-Leibler divergence KL (performance measure) versus the number of nodes $N$, comparing the edgemcd method proposed in this paper in black with the standard std estimates in red and the mcd method in green, see text, for two different dimensions, solid ($p=3$) and dashed ($p=10$) lines, with varying corruption level in each column and different types of graph in each row.}
    \label{fig::kl}
    \centering
    \includegraphics[scale = 0.28]{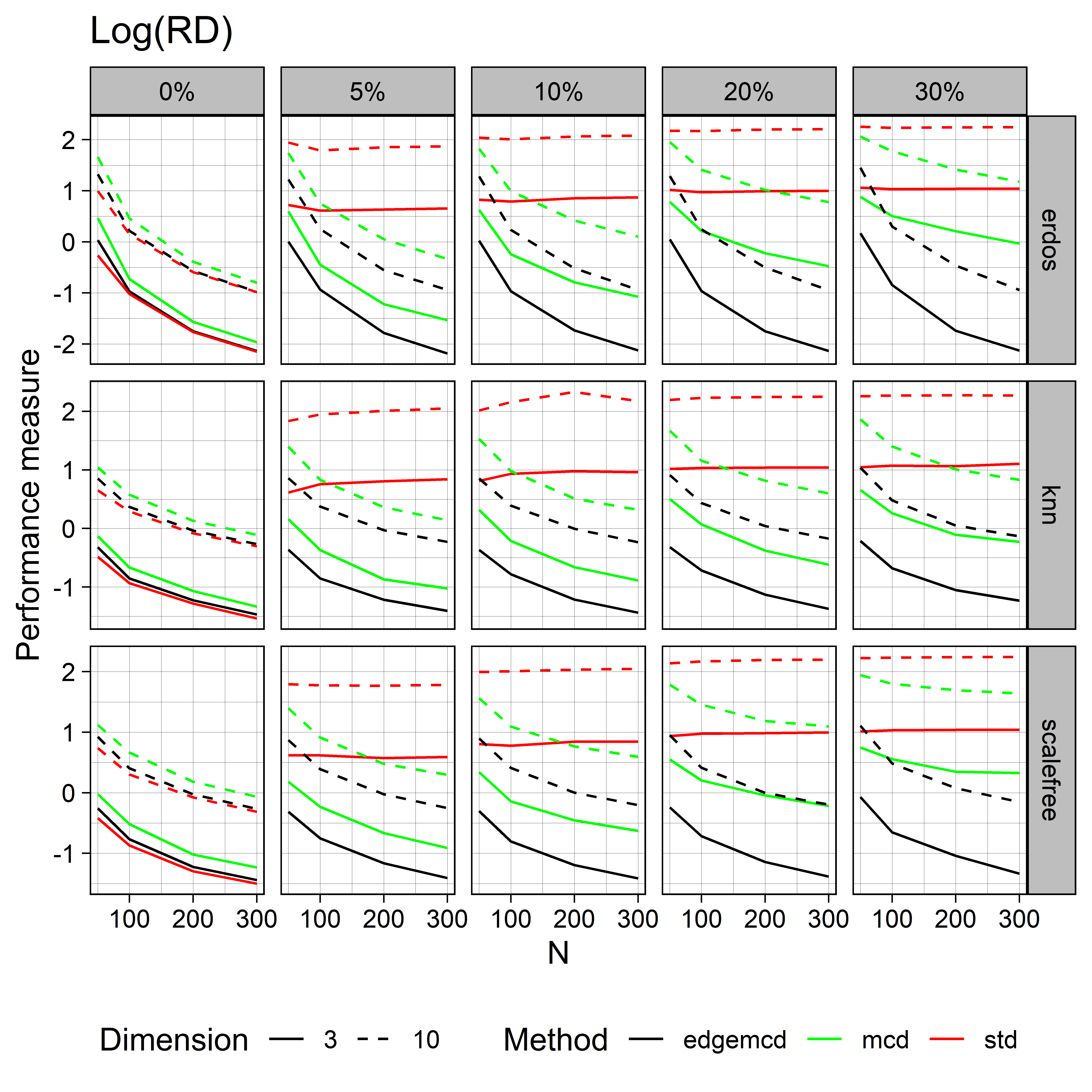}
    \caption{Log transformed relative distance RD (performance measure) versus the number of nodes $N$, comparing the edgemcd method proposed in this paper in black with the standard std estimates in red and the mcd method in green, see text, for two different dimensions, solid  ($p=3$) and dashed ($p=10$) lines, with varying corruption level in each column and different types of graph in each row.}
    \label{fig::theta}
\end{figure}

\begin{figure}[h]
    \centering
    \includegraphics[scale = 0.3]{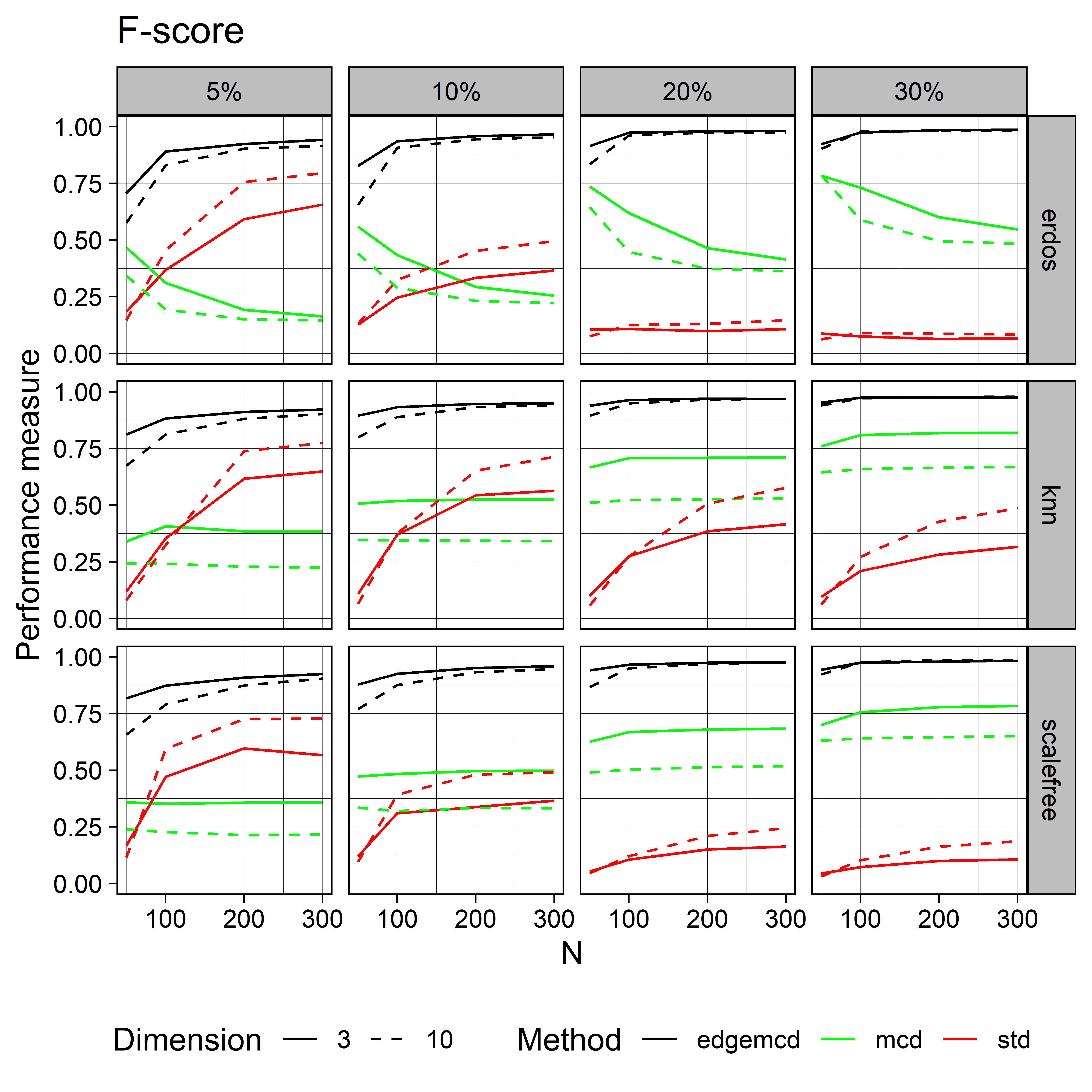}
    \caption{F-score (performance measure) versus the number of nodes N, comparing the method proposed in this paper (edgemcd) in black with the standard estimates (std) in red and another robust estimator (mcd) in green, see text, for different dimensions, solid and dashed lines, with varying corruption level in each column and different types of graph in each row.}
    \label{fig::fs}
\end{figure}


\section{Electoral Data}

We use the model and the edgewise outlier detection method as described in subsection \ref{subsec:edgemcd} to analyze an election  dataset publicly available at \url{https://www.data.gouv.fr/fr/datasets/elections-departementales-2015-resultats-par-bureaux-de-vote/}. The data contains vote shares   (in percent) for three groups of political parties  $\vec{x}_i \in \mathbb{R}^3$ (left parties, right parties, others)  for 95 French departments for the 2015 departmental elections. The covariates $\vec{z}_i \in \mathbb{R}^{19}$ are the population share  in age groups  (between 18-39 (ag\_1839), 40-64 (ag\_4064), or above 65 (age\_65)), 
the population share in  employment categories (agriculture and fisheries (AZ), manufacturing
industry, mining industry, and others (BE), construction (FZ), business, transport
and services (GU) and public administration, teaching and human health (OQ)), the proportion of foreigners (foreign), the proportion of income taxpayers (incm\_rt), the proportion of asset owners (ownr\_rt), the unemployment rate (unmp\_rt), the mean annual employment growth (emply\_v), and the number of people with different education levels (secondary (N\_CAPBE), at least secondary and most high school (bac) and  a university degree (dplm\_sp)).

As each datapoint $\vec{x}_i $ at a location consists of the percentage of voters for each of the three categories, it would be inappropriate to deal with this data in a Euclidean way. Similarly, some covariate groups such as voter age class, employment type, and education level are better interpreted in terms of percentages than absolute numbers. In fact \cite{nguyen2022analyzing} treat this data set as compositional. As some of the readers might be unfamiliar with compositional data we introduce the main concepts quickly. 

\subsection{Compositional Data}

Compositional data consists of strictly multivariate positive data and is easiest thought of as being restricted by $\sum_{k = 1}^p x_k =1$. The set of such vectors is called the $p$-part simplex 
\begin{align*}
    \mathcal{S}^p:=\bigg\{ (x_1,\ldots,x_p)^\top \in \mathbb{R}^p_{+} : \sum_{k = 1}^p x_k =1 \bigg\} \subset \mathbb{R}^p_+ \ 
\end{align*}
and it is equipped with an addition and a multiplication operation, also called perturbation and powering in the compositional literature defined as    
$$\vec{x} \oplus_{\mathcal{A}} \vec{y} := \frac{1}{ \sum_{i=1}^p x_i y_i}(x_1 y_1,\ldots,x_p y_p)^\top$$ 
and 
$$\alpha \odot_{\mathcal{A}} \vec{x} :=  \frac{1}{ \sum_{i=1}^p x_i^\alpha} (x_1^\alpha,\ldots,x_p^\alpha)^\top,$$
for any $\vec{x} \in \mathcal{S}^p $ and $\vec{y} \in \mathcal{S}^p$, and $\alpha \in \mathbb{R}$, see \cite{pawlowsky2015modeling}. In addition to the perturbation and powering operation, an inner product can be defined
 \begin{align}
        {\langle \vec{x},\vec{y} \rangle}_{\mathcal{A}} := \frac{1}{2p} \sum^p_{k,l = 1} \ln\bigg(\frac{x_k}{x_l}\bigg)\ln\bigg(\frac{y_k}{y_l}\bigg) , \label{aitchisonInner}
\end{align}
turning $(\mathcal{S}^p, {\langle \cdot , \cdot \rangle}_{\mathcal{A}} , \oplus_{\mathcal{A}}, \odot_{\mathcal{A}})$ into a finite $p-1$ dimensional Hilbert space with norm $\norm{\vec{x}}_{\mathcal{A}} := \sqrt{{\langle \vec{x},\vec{x} \rangle}_{\mathcal{A}}}$, see \cite{pawlowsky2015modeling}. Two transformations are of central importance in compositional data analysis. The first one is the
clr (centered log-ratio)-map given as
\begin{align} \label{standardCLR}
  \clr : \mathcal{S}^p \rightarrow \mathbb{R}^p, \quad    \clr{(\vec{x})} := \Bigg(\ln\Bigg(\frac{x_1}{\sqrt[p]{\prod^p_{k=1} x_k}}\Bigg),\ldots ,\ln\Bigg(\frac{x_p}{\sqrt[p]{\prod^p_{k=1} x_k}}\Bigg)\Bigg)^\top ,
\end{align}
which along with being distance preserving (see  \cite{pawlowsky2015modeling}) also fulfills $\clr(\vec{x} \oplus_{\mathcal{A}} \vec{y}) = \clr(\vec{x}) + \clr(\vec{y})$, $\clr(\alpha \odot_{\mathcal{A}} \vec{x}) = \alpha \clr(\vec{x}) $ and ${\langle \vec{x},\vec{y} \rangle}_{\mathcal{A}} = {\langle \clr(\vec{x}),\clr(\vec{y}) \rangle}_{2}$. However, the clr-map is not bijective onto $\mathbb{R}^p$ and therefore a more useful map, called the ilr (isometric log-ratio)-map, see \cite{egozcue2003isometric}, is given by
\begin{align} \label{ilrCoda}
   \ilr_{\matr{V}} : \mathcal{S}^p \rightarrow \mathbb{R}^{p-1}, \quad  \ilr_\matr{V}(\vec{x}) := \matr{V}'\clr(\vec{x}) \ ,
\end{align}
where  $\matr{V} \in \mathbb{R}^{p \times (p-1)}$ is a matrix with orthogonal columns spanning the $p-1$ dimensional subspace $\{ \vec{a} \in \mathbb{R}^p:\sum_{j=1}^p a_j = 0 \} \subset \mathbb{R}^D$. The ilr-map is an isometric bijective map onto $\mathbb{R}^{p-1}$ and its foremost advantage is to transform compositional data to the standard Euclidean geometry where standard methods can be used. To transform a point back to the simplex we can simply use the following relation $\matr{V} \ilr_\matr{V}(\vec{x}) = \clr(\vec{x}) $.

\subsection{Electoral Data}

Following the previous subsection we apply the clr-transformation, and the ilr-transformation, to each row of $\matr{X}$ as well as to the covariates population age distribution, employment distribution, and education level. The other covariates remain unchanged. Denote the resulting data matrices by $\matr{X}^{\clr}$, $\matr{Z}^{\clr}$, $\matr{X}^{\ilr}$ and $\matr{Z}^{\ilr}$. By properties of the ilr-transform (equation \eqref{ilrCoda}), we can write $\vec{x}^{\ilr}_i = \matr{V}_{\matr{X}}' \vec{x}^{\clr}_i$ and $ \vec{z}_i^{\ilr} = \matr{V}_{\matr{Z}}'\vec{z}^{\clr}_i$ for some matrices $\matr{V}_{\matr{X}} $ and $\matr{V}_{\matr{Z}}$.
We apply the algorithm discussed in subsection \ref{subsec:edgemcd} to $\matr{X}^{\ilr}$ and $\matr{Z}^{\ilr}$ to find robust estimators $(\hat{\vec{\theta}}^{\ilr},\hat{\matr{\Sigma}}_V^{\ilr})$. We can rewrite $\vec{x}^{\ilr,\hat{\vec{\mu}}^{\ilr}}_i=\vec{x}^{\ilr}_i-\vec{z}^{\ilr}_i \hat{\vec{\theta}}^{\ilr}$ as 
\begin{align*}
\vec{x}^{\ilr}_i-\vec{z}^{\ilr}_i \hat{\vec{\theta}}^{\ilr} = \matr{V}_{\matr{X}}' \vec{x}^{\clr}_i - (\hat{\vec{\theta}}^{\ilr})' \matr{V}_{\matr{Z}}'\vec{z}^{\clr}_i =   \matr{V}_{\matr{X}}' \vec{x}^{\clr}_i - (\matr{V}_{\matr{Z}}\hat{\vec{\theta}}^{\ilr} )'\vec{z}^{\clr}_i  
&= \matr{V}_{\matr{X}}'( \vec{x}^{\clr}_i - \matr{V}_{\matr{X}}(\matr{V}_{\matr{Z}}\hat{\vec{\theta}}^{\ilr} )'\vec{z}^{\clr}_i) \\
&= \matr{V}_{\matr{X}}'( \vec{x}^{\clr}_i - (\hat{\pmb{\theta}}^{\clr})'\vec{z}^{\clr}_i) \\
& = \matr{V}_{\matr{X}}' \vec{x}^{\clr,\hat{\vec{\mu}}^{\clr}}_i
\end{align*}
where we set $\hat{\pmb{\theta}}^{\clr} :=\matr{V}_{\matr{Z}}\hat{\vec{\theta}}^{\ilr} \matr{V}_{\matr{X}}'$ and $\hat{\pmb{\mu}}^{\clr}:= \matr{Z}^{\clr} \hat{\pmb{\theta}}^{\clr} $.  Consequently, by defining $\matr{\Sigma}_{V}^{\clr} 
:= \matr{V}_{\matr{X}}\hat{\matr{\Sigma}}_V^{\ilr} \matr{V}_{\matr{X}}'$, we can also rewrite 
\begin{align*}
    {\hat{\Delta}_{ij}}^{\ilr} &=(\vec{x}^{\ilr,\hat{\vec{\mu}}^{\ilr}}_i-  \vec{x}^{\ilr,\hat{\vec{\mu}}^{\ilr}}_j)'(\hat{\matr{\Sigma}}_V^{\ilr})^{-1} (\vec{x}^{\ilr,\hat{\vec{\mu}}^{\ilr}}_i- \vec{x}^{\ilr,\hat{\vec{\mu}}^{\ilr}}_j)w_{ij} \\
    & =  (\vec{x}^{\clr,\hat{\vec{\mu}}^{\clr}}_i-\vec{x}^{\clr,\hat{\vec{\mu}}^{\clr}}_j)'(\matr{V}_{\matr{X}}(\hat{\matr{\Sigma}}_V^{\ilr})^{-1} \matr{V}_{\matr{X}}')(\vec{x}^{\clr,\hat{\vec{\mu}}^{\clr}}_i-\vec{x}^{\clr,\hat{\vec{\mu}}^{\clr}}_j)w_{ij} \\
    & = (\vec{x}^{\clr,\hat{\vec{\mu}}^{\clr}}_i-\vec{x}^{\clr,\hat{\vec{\mu}}^{\clr}}_j)'(\matr{\Sigma}_{V}^{\clr})^{+}(\vec{x}^{\clr,\hat{\vec{\mu}}^{\clr}}_i-\vec{x}^{\clr,\hat{\vec{\mu}}^{\clr}}_j)w_{ij}
\end{align*}
where we use for the last equation that the generalized inverse of $\matr{V}_{\matr{X}}\hat{\matr{\Sigma}}_V^{\ilr} \matr{V}_{\matr{X}}'$ is given by $\matr{V}_{\matr{X}}(\hat{\matr{\Sigma}}_V^{\ilr})^{-1} \matr{V}_{\matr{X}}'$.
We see that ${\hat{\Delta}_{ij}}^{\ilr}$ does not depend on the contrast matrix $\matr{V}_{\matr{X}}$ and therefore can also be denoted by  ${\hat{\Delta}_{ij}}^{\clr}.$ We then see that in clr-coordinates an edge is an outlier if ${\hat{\Delta}_{ij}}^{\clr} > \chi^2(p-1,0.995)$. Additionally, $\matr{\Sigma}_{V}^{\clr}$ can be interpreted as the global covariance between the different voter shares and $\hat{\pmb{\theta}}^{\clr}$ as the coefficients corresponding to a covariance variable driving the voter share results.

Figure \ref{fig::franceNetwork} shows the network structure that we consider. The departments that share a border are connected. 
Denote $\matr{W} \in \mathbb{R}^{93 \times 93}$ a weight matrix that has entries ${w}_{ij}$ equal to one if there is an edge between department $i$ and department $j$ and zero otherwise. We then set the Laplacian matrix to $\matr{L} = \matr{D} - \matr{W}$, see section \ref{sec:graphsig}.

\begin{figure}[h]
\centering
\includegraphics[scale = 0.9]{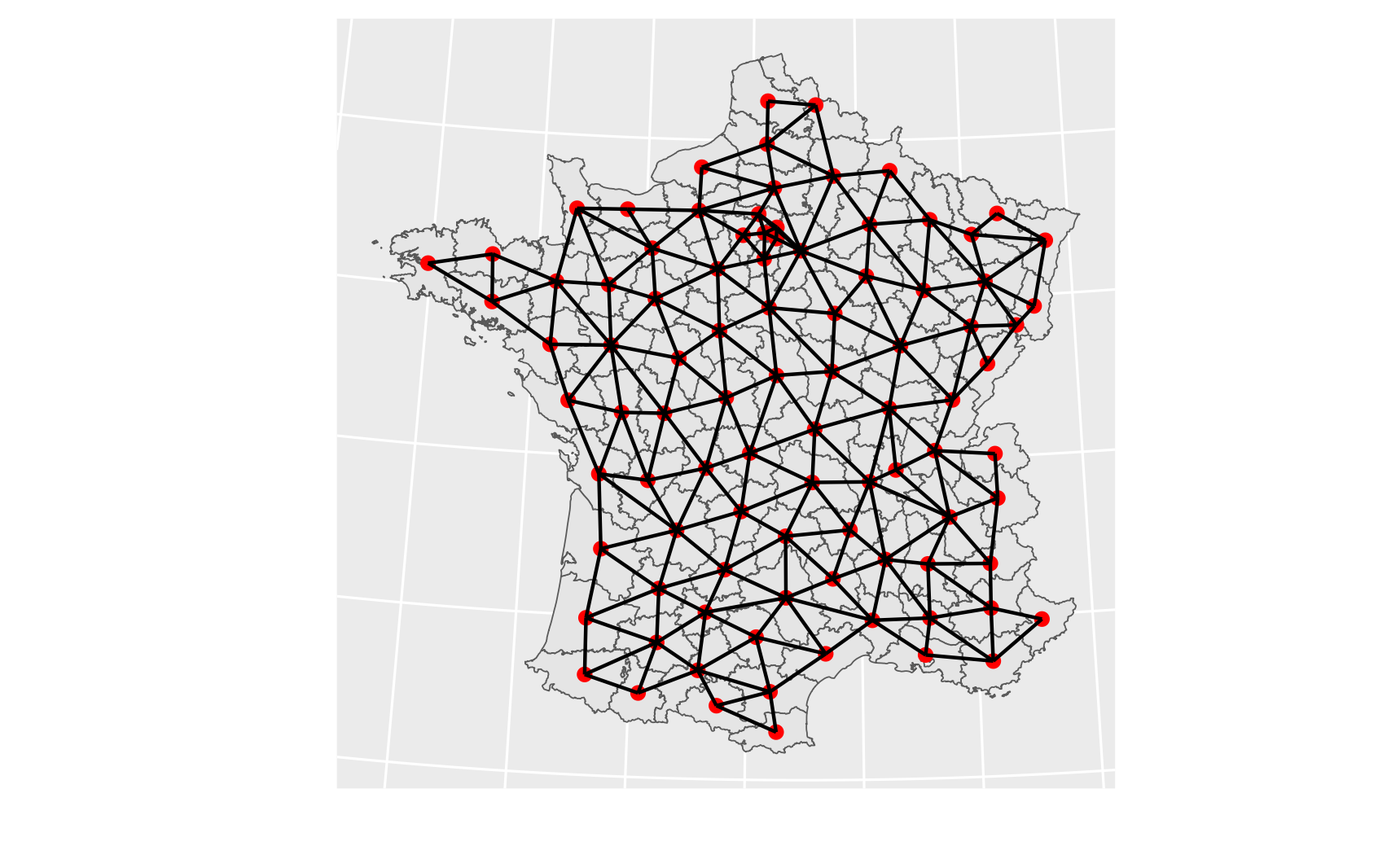}
\caption{ Network connecting adjacent departments of France. }
\label{fig::franceNetwork}
\end{figure}

 Table \ref{table::coeffs} displays  the estimated coefficients $\hat{\pmb{\theta}}^{\clr} $ in clr-coordinates normalized by multiplying by the standard deviation of each covariate. Voters with lower education levels (N\_CAPBE) tend to vote less for left parties and more for right and others. On the contrary, voters with a university degree (dplm\_sp) vote similarly for the left and right but much less for others. An increase of the unemployment rate (unmp\_rt) leads to more votes for the parties in the category others than for the right parties. Similarly, an increase in the rate of foreigners leads to more votes for parties contained in others. Increasing the percentage of employment in the construction sector (FZ) leads to more votes for the right parties.

\begin{table}[ht]
\centering
\begin{tabular}{rrrrrrrrrrr}
  \hline
 & ag\_1839 & ag\_4064 & age\_65 & N\_CAPBE & bac & dplm\_sp & AZ & BE & FZ & GU \\ 
  \hline
left & -0.03 & -0.00 & 0.05 & -0.24 & 0.05 & 0.09 & -0.05 & 0.01 & -0.11 & 0.02 \\ 
  right & -0.09 & 0.02 & 0.05 & 0.11 & -0.05 & 0.05 & 0.04 & -0.05 & 0.10 & -0.04 \\ 
  others & 0.13 & -0.01 & -0.09 & 0.13 & -0.00 & -0.13 & 0.01 & 0.04 & 0.01 & 0.02 \\ 
   \hline
\end{tabular}
\begin{tabular}{rrrrrrr}
  \hline
 & OQ & unmp\_rt & emply\_v & ownr\_rt & incm\_rt & foreign \\ 
  \hline
left & 0.09 & 0.01 & -0.03 & 0.09 & -0.13 & -0.05 \\ 
  right & -0.04 & -0.22 & 0.01 & -0.19 & 0.04 & -0.07 \\ 
  others & -0.05 & 0.21 & 0.02 & 0.10 & 0.09 & 0.12 \\ 
   \hline
\end{tabular}
\caption{Standardized estimated coefficients}
\label{table::coeffs}
\end{table}

Figure \ref{fig::mhdsplot1} shows $\sqrt{\frac{{\hat{\Delta}_{ij}}^{\clr}}{ w_{ij}(l_{ii}+l_{jj}-2l_{ij})}}$ on the y-axis versus the edges $(i,j) \in \mathcal{E}$ on the x-axis. Everything above the horizontal line at $\sqrt{\chi^2(2,0.975)}$ can be considered as an edge outlier. This map is helpful for checking for edge outliers and their magnitude. Clearly, we can see that there are couples of adjacent departments that display very different behavior. However, as the indexing of the edges on the x-axis is arbitrary this plot is  helpful in detecting these departments but needs to be completed by a corresponding map. Figure \ref{fig::mhdsFranceParis} shows these outlying edges for the whole of France, where the darker an edge is the more it is outlying , i.e.  $\sqrt{\frac{{\hat{\Delta}_{ij}}^{\clr}}{ w_{ij}(l_{ii}+l_{jj}-2l_{ij})}}$ is comparatively bigger. 
When looking at the whole country, we can see that certain departments behave very differently from their neighbors. We will only look at the biggest outliers. In the south-west of France the departments Lot, Corrèze, and Cantal show high edge outliers. A possible explanation can be found by looking at the values of $\vec{x}_i$ and $\vec{z}_i$ for this region. 
We look at the log-ratios of $\vec{x}_i$ in that region, as is common in Compositional Data, see top row of Figure \ref{fig::lrall}. We can see that voters in Lot voted primarily for the left party rather than the right or others. Cantal was primarily dominated by the right party. Votes in Corr\`eze were almost equally split between the left and right parties. The map for the log-ratio between left and others displays little spatial change and we can make the likely conclusion that the edge outliers in this area were caused by the domination of the left in Lot and the right in Cantal. 
No atypical values of $\vec{z}_i$ seem to drive these outliers except possibly that the agriculture and fisheries (AZ) sector takes a much bigger role in Cantal than in the other departments. 
The higher the (AZ) sector is the lower the votes for the left are, as can be seen in Table \ref{table::coeffs}, which might explain the domination of the right in Cantal. Similarly, Figure \ref{fig::mhdsFranceParis} shows an outlying edge between the departments of Ariège and Pyrénées-Orientales. Again looking at the log-ratio maps, middle row of Figure \ref{fig::lrall}, we can make the likely conclusion that the outliers are caused by the high percentage of votes for the left in Ari\`ege and the almost equal split of votes between the three voter categories in Pyrénées-Orientales. This is rather atypical as the surrounding departments leaned rather to the left.  
Again, no atypical values of $\vec{z}_i$ seem to be the reason except for the sector industry, mining industry, and others (BE) taking a bigger part in Ari\`ege. Finally, zooming into I\^le de France, we see that the biggest edge outliers are found for the departments of \^Ile-de-France, see the right plot of Figure \ref{fig::mhdsFranceParis}. The edgewise Mahalanobis distance $\sqrt{\frac{{\hat{\Delta}_{ij}}^{\clr}}{ w_{ij}(l_{ii}+l_{jj}-2l_{ij})}}$ between the department Seine-Saint-Denis and respectively Hauts-de-Seine, Seine-et-Marne, Val-de-Marne and Val d'Oise (in order of magnitude)  are especially high. Again Figure \ref{fig::lrall} helps us gain insight into the reason for these outliers. Seine-Saint-Denis was heavily dominated by the left party whereas its neighbor Hauts-De-Seine was mainly dominated by the right party. Seine-Et-Marne and Val-D'Oise votes were almost equally split which might explain the outliers between the latter two and Seine-Saint-Denis a heavily left department. There are some atypical values of $\vec{z}_i$ that might be responsible for these edge outliers. Mainly the department of Seine-Saint-Denis has a comparatively high number of inhabitants above 65 (age\_65) and also a high number of voters with only secondary education (N\_CAPBE). 
Typically, the latter would be associated with fewer voters for the left, see Table \ref{table::coeffs}, contradicting the high percentage of voters for the left in this department. Also Hauts-De-Seine has a very low part of voters with only secondary education (N\_CAPBE), in fact, the lowest for the whole of France.
No node outliers were detected for this data set.

\begin{figure}[h]
\centering
\includegraphics[scale = 0.75]{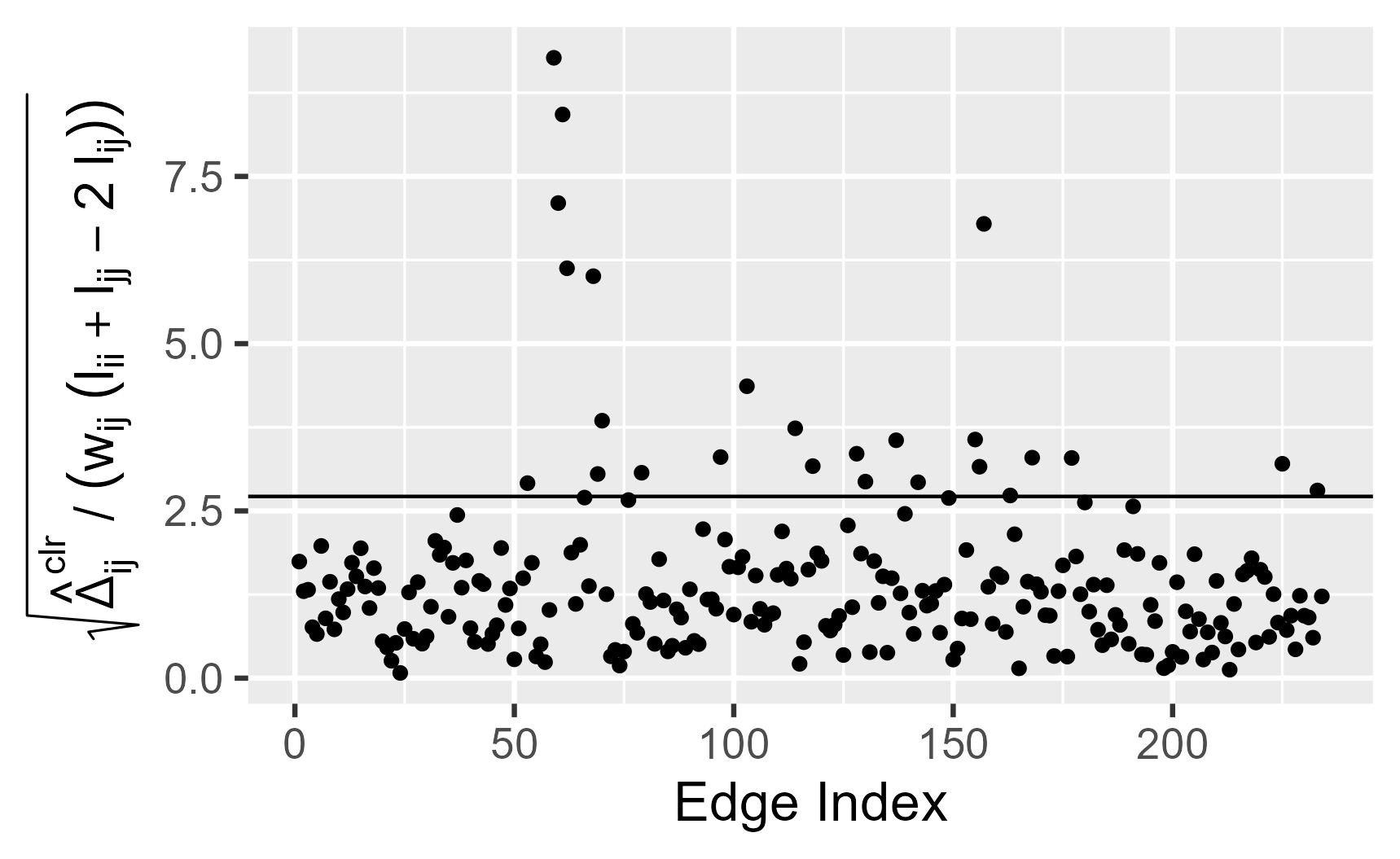}
\caption{Plot of the square root of $\frac{{\hat{\Delta}_{ij}}^{\clr}}{ w_{ij}(l_{ii}+l_{jj}-2l_{ij})}$ on the y-axis versus the edges $(i,j) \in \mathcal{E}$ on the x-axis. The horizontal line is at $\sqrt{\chi^2(2,0.975)}$. Potential edge outliers are above the line.}
\label{fig::mhdsplot1}
\end{figure}

\begin{figure}[h]
\centering
\includegraphics[width=1\textwidth]{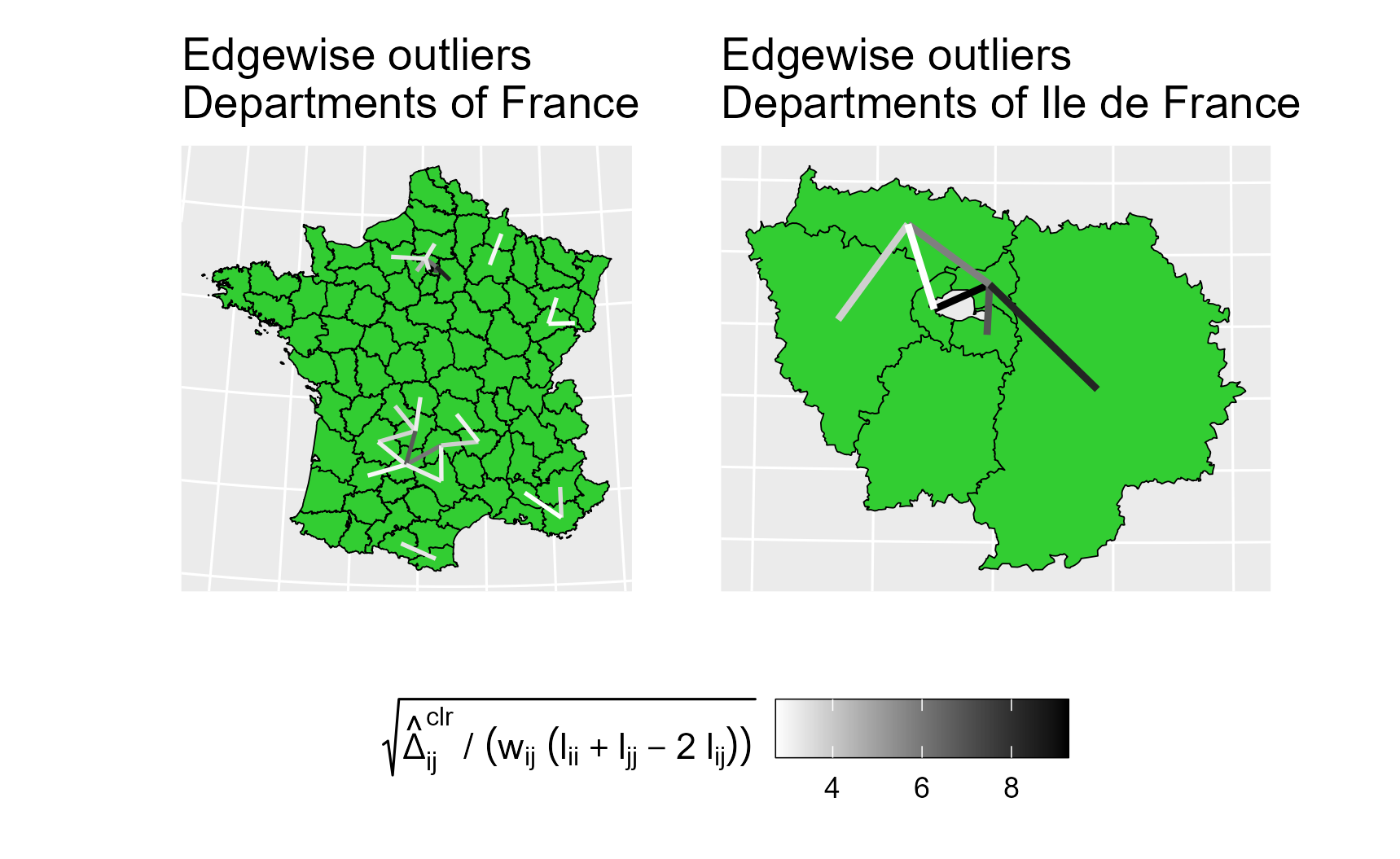}
\caption{On the left a map of mainland France with its departments with the detected edgewise outliers from white to black depending on their strength of outlyingness $\frac{\Delta_{ij}}{ w_{ij}(l_{ii}+l_{jj}-2l_{ij})}$. On the right a zoom into Paris with central Paris missing due to NAs.}
\label{fig::mhdsFranceParis}
\end{figure}

\begin{figure}[h]
  \centering
  \includegraphics[width=1\textwidth]{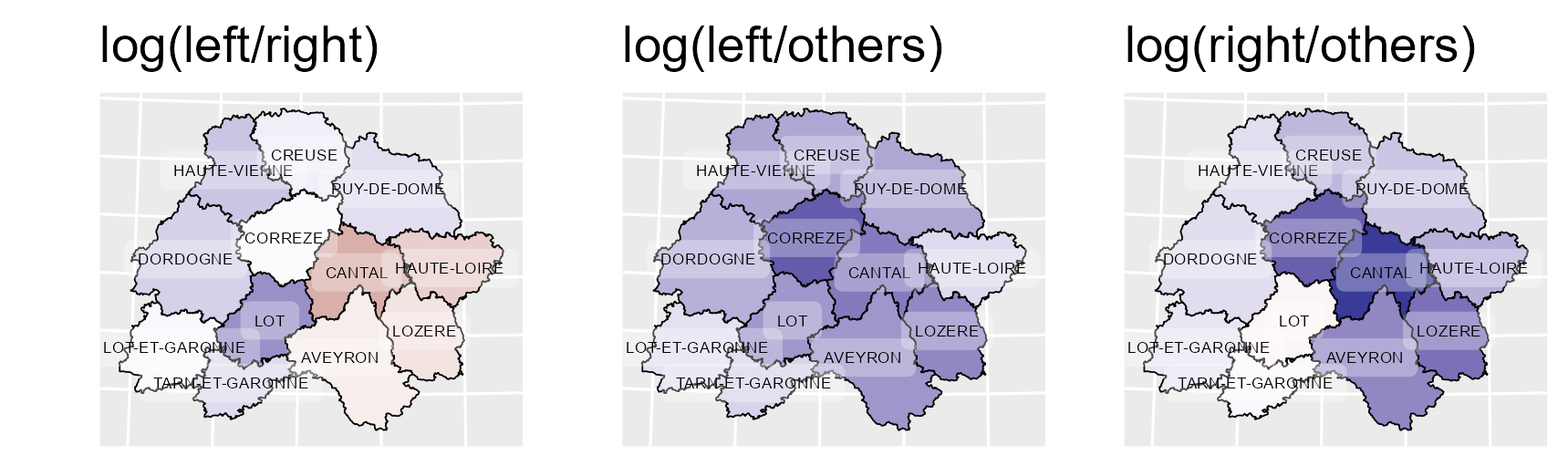}
  \vspace{-10pt}  
  \includegraphics[width=1\textwidth]{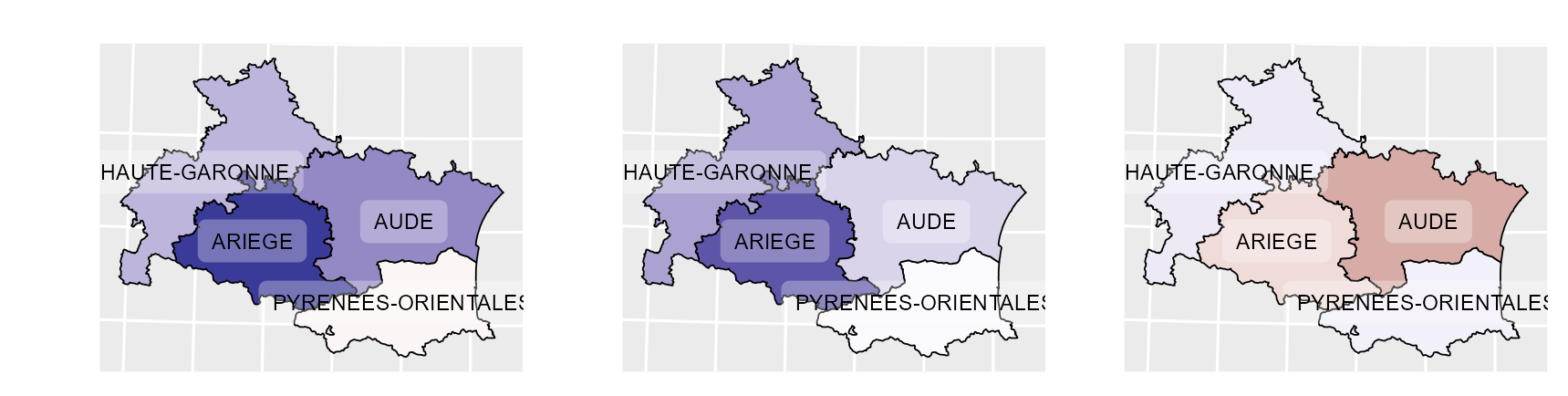}
  \vspace{-10pt}  
  \includegraphics[width=1\textwidth]{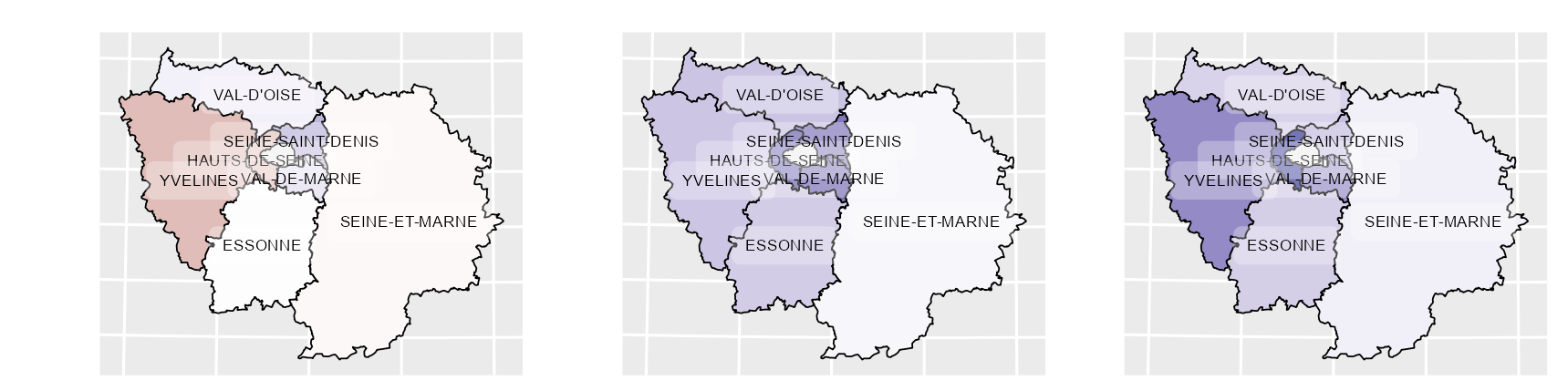}
  \caption{Log-ratio maps of departments with high edge outliers including their neighbors for better comparison. Red is negative and blue positive. The darker a color the higher the log-ratio.}
  \label{fig::lrall}
\end{figure}

To visually check the validity of the proposed model, we can also look at the standardized residuals $\matr{L}^{\frac{1}{2}} (\matr{X}^{\clr} - \matr{Z}^{\clr}\hat{\pmb{\theta}}^{\clr} ) (\hat{\matr{\Sigma}}_{V}^{\clr})^{-\frac{1}{2}}$ for each department and each response (left, right, other). If the model (\ref{eq:model}) holds for $\matr{X}^{\clr}$ and $\matr{Z}^{\clr}$ then we would expect the standardized residuals to behave like white noise up to a rotation. Even though our data does not necessarily follow (\ref{eq:model}), due for example to outliers, it is still reasonable to look at the residuals for quick model checking. Figure \ref{fig::residuals3} shows almost no spatial patterns (the left plots) except for pairs of departments that have been detected as edgewise outliers. On the right side, the plots which show the residuals versus the node index, also display little unusual behavior except for some departments such as Seine-Saint-Denis visible in the right bottom plot of Figure \ref{fig::residuals3}.

\begin{figure}[h]
\centering
\includegraphics[width=0.7\textwidth]{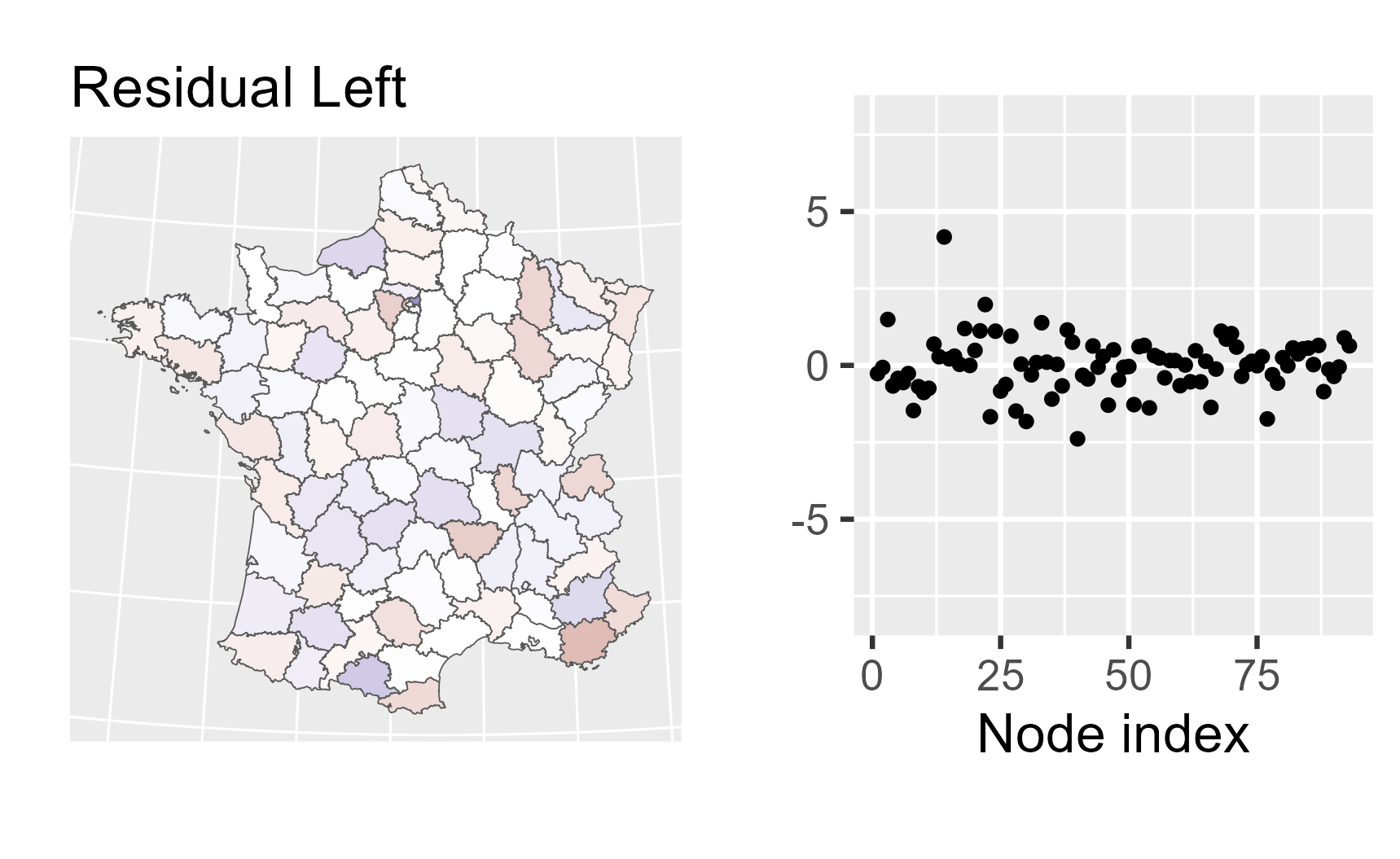}

\vspace{-20pt}  

\includegraphics[width=0.7\textwidth]{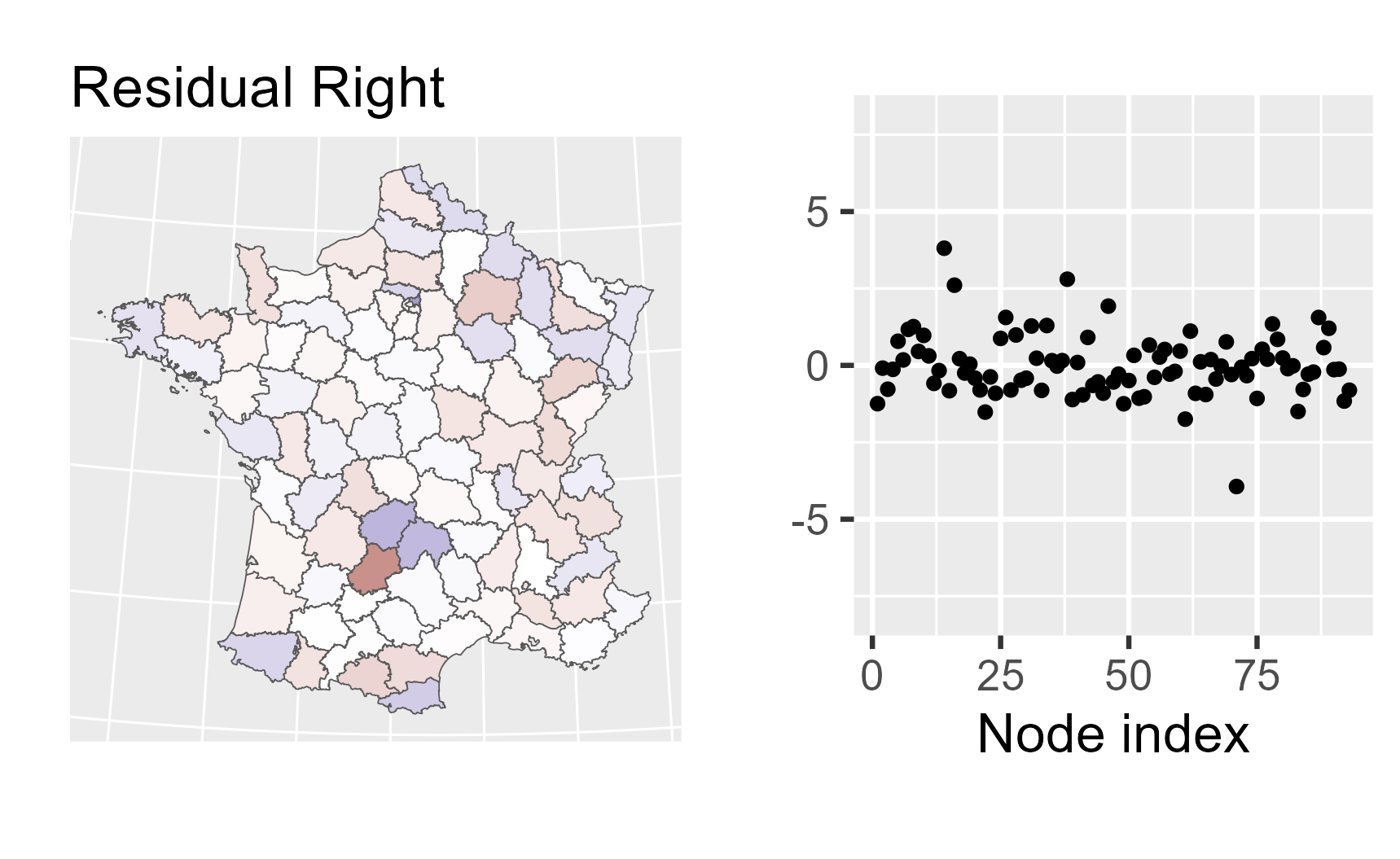}

\vspace{-20pt}  

\includegraphics[width=0.7\textwidth]{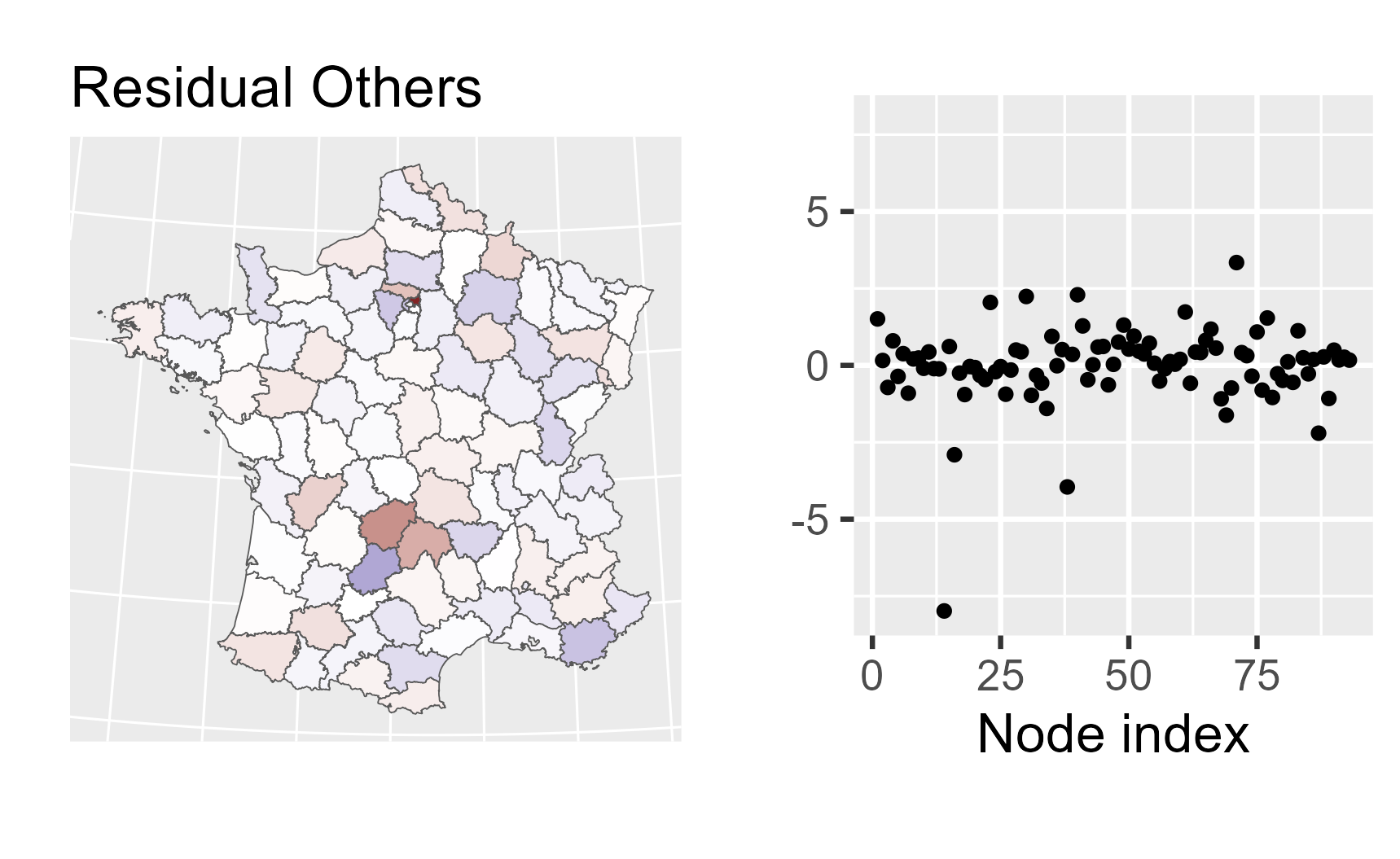}
\caption{ On the left, the standardized residuals $\matr{L}^{\frac{1}{2}} (\matr{X}^{\clr} - \matr{Z}^{\clr}\hat{\pmb{\theta}}^{\clr} ) (\hat{\matr{\Sigma}}_{V}^{\clr})^{-\frac{1}{2}}$ for each category of voters (left, right, others) for each department. Blue is positive and red is negative. The right plots show the same but are dependent on the node index. The residuals have been scaled by the maximum to make the plots comparable.}
\label{fig::residuals3}
\end{figure}

\section{Conclusion}
The current literature contains few proposals for modeling graph-indexed data in the univariate setting. In this paper, we consider Gaussian models for multivariate graph-indexed data taking into account the network dependence as well as the dependence of the variables. The Mahalanobis distance is frequently used for outlier detection but, up to our knowledge, there has not yet been a development for dependent, particularly graph-indexed, data. We introduce a new concept that we dubbed edgewise outliers. That is, given a graph structure, with multivariate data indexed by the nodes, we find edges such that incident data points are very dissimilar. We formulate decision rules for the detection of such edgewise outliers in the framework of the proposed model. We introduce a robust estimation method for their parameters inspired by the deterministic MCD algorithm. Our simulations show that the edgewise MCD algorithm outperforms the classical MCD and the standard  maximum likelihood for different performance measures including estimators quality as well as outlier detection quality. Finally, we also show the utility of our method on the French departmental election data of 2015 finding neighbouring departments that behave unalike. A setting not covered in this paper but left for future research is the high dimensional setting when the number of nodes is much smaller than the number of variables.

\section*{Statements and Declarations}

The authors declare that they have no conflict of interest.

\section*{Acknowledgements}

This work was supported by the Austrian Science Fund (FWF) under grant P32819 Einzelprojekte and grant P31881-N32.

\section*{Code}

All computations of this paper were done in the R-programming language with the core written in C++ for performance reasons. The code is available at the GitHub repository \url{https://github.com/Kristats/SpOut.git}.

%
%
\FloatBarrier

\begin{appendices}

\section{Proofs} \label{appB}

\begin{proof}[Proof of Theorem \ref{thm:mattrans}]
As the vectorization of $\matr{X}$ is a Gaussian random variable so is any linear combination of the latter, especially $\matr{A} \matr{X} \matr{B}$. Because taking the expectation is a linear operation we have $\E(\matr{A} \matr{X} \matr{B}) = \matr{A} \E(\matr{X}) \matr{B} = \matr{A}  \pmb{\mu} \matr{B}$. By using twice property (\ref{eq:covProd}) the covariance of two entries of the latter is equal to 
\begin{align*}
    \Cov((\matr{A} \matr{X} \matr{B})_{ik},(\matr{A} \matr{X} \matr{B})_{jl}) &= \sum_{s,r=1}^n \sum_{t,d=1}^p \Cov(a_{is} x_{st} b_{tk},a_{jr} x_{rd} b_{dl}) \\
    &=  \sum_{s,r=1}^n \sum_{t,d=1}^p a_{is}b_{tk}a_{jr}b_{dl}  \Cov( x_{st} , x_{rd} )  \\
    &= \sum_{s,r=1}^n \sum_{t,d=1}^p a_{is}b_{tk}a_{jr}b_{dl}   (\matr{\Sigma}_G)_{sr}(\matr{\Sigma}_V)_{td} \\
    &= \bigg(\sum_{s,r=1}^n  a_{is} (\matr{\Sigma}_G)_{sr} a_{jr}  \bigg) \bigg( \sum_{t,d=1}^p b_{tk} (\matr{\Sigma}_V)_{td} b_{dl} \bigg) \\
    &= (\matr{A} \matr{\Sigma}_G \matr{A}')_{ij}(\matr{B}' \matr{\Sigma}_V \matr{B})_{kl},
\end{align*}
which concludes the proof.
\end{proof}

\begin{proof}[Proof of Lemma \ref{lem:maha}]
    By properties of the vectorization operator and the Kronecker product (see \cite{harville1998matrix}), in particular $\vecc(\matr{AXB})=(\matr{B}'\otimes \matr{A})\vecc(\matr{X}) $, we have 
    \begin{eqnarray}
    \md2(\matr{X}) & = & 
    \norm{(\matr{\Sigma}_V \otimes \matr{L}^+)^{+/2}\vecc(\matr{X}-\pmb{\mu})}^2 \nonumber\\
     & = & \norm{(\matr{\Sigma}_V^{-1/2} \otimes \matr{L}^{1 /2})\vecc(\matr{X}-\pmb{\mu})}^2\nonumber\\
     \md2(\matr{X}) & = &  \norm{ \matr{L}^{1/2}(\matr{X}-\pmb{\mu})\matr{\Sigma}_V^{-1/2}}_F^2, \label{md_frob}
    \end{eqnarray}
    where $\norm{A}_F$ is the Frobenius norm of matrix $\matr{A}.$
    By properties of the trace $\Tr$ we can write 
    \begin{align*}
        \norm{ \matr{L}^{1/2}(\matr{X}-\pmb{\mu})\matr{\Sigma}_V^{-1/2}}^2_{F} = \Tr((\matr{X}^{\pmb{\mu}}\matr{\Sigma}_V^{-1/2})'\matr{L}(\matr{X}^{\pmb{\mu}}\matr{\Sigma}_V^{-1/2})).
    \end{align*}
Setting $\matr{Z}:= \matr{X}^{\pmb{\mu}}\matr{\Sigma}_V^{-1/2}$ we further have $
\Tr(\matr{Z}'\matr{L}\matr{Z}) = \sum_{k=1}^p \vec{z}_{:,k}' \matr{L} \vec{z}_{:,k},$ where $\vec{z}_{:,k}$ denotes the $k^{th}$ column of matrix $\matr{Z}$. By properties of the Laplacian matrix, see \cite{merris1994laplacian}, we know that each summand $\vec{z}_{:,k}' \matr{L} \vec{z}_{:,k}$ is equal to $\frac{1}{2}\sum_{i,j=1}^p (z_{ik} - z_{jk})^2 w_{ij} $. Therefore we get $\Tr(\matr{Z}'\matr{L}\matr{Z}) = \frac{1}{2}\sum_{i,j=1}^p \norm{\vec{z}_{i,:} - \vec{z}_{j,:}}^2 w_{ij},$ where $\vec{z}_{i,:}$ denotes the $i^{th}$ row of matrix $\matr{Z}$. Substituting for $\matr{Z}$ we get the result.
\end{proof}

\begin{proof}[Proof of Lemma \ref{lem:deltadistr}]
For $i=1,\dots,n$, let us define the vectors $\vec{e}_i$ that have zero components except at position $i$ where the component is $1.$ To derive the distribution of $\Delta_{ij}$ we first note that $\vec{e}_{ij}:= \vec{e}_{i} - \vec{e}_{j}\in \mathbb{R}^{n}$, that is zero except at position $i$ where it is $1$ and $-1$ at position $j$, i.e. $\vec{e}_{ij}=(0,\dots,0,1,0,\dots,0,-1,0,\dots,0)'$, satisfies 
\begin{align*}
    \vecc{(\vec{e}_{ij}'\matr{X})} \sim \mathcal{N}_{np} ( \vecc{(\vec{e}_{ij}'\pmb{\mu})}, \matr{\Sigma}_V \otimes  \vec{e}_{ij}'\matr{L}^{+} \vec{e}_{ij}).
\end{align*}
From this it is easy to deduce $ (\vec{x^{\mu}}_i - \vec{x^{\mu}}_j)\matr{\Sigma}_V^{-1/2} \sim \mathcal{N}_{np} ( \vec{0}, \matr{I}\sigma_{ij}^2) $ with $\sigma_{ij}^2=  \vec{e}_{ij}'\matr{L}^{+} \vec{e}_{ij} = l_{ii}+l_{jj}-2l_{ij}$. Thus $\displaystyle \frac{\Delta_{ij}}{w_{ij} \sigma_{ij}^2}  \sim \chi^2(p)$.
%

\end{proof}

\begin{proof}[Proof of Theorem \ref{thm::likelihoodestimators}]
    Using equation \eqref{md_frob}, the negative log-likelihood of the model can be written, where we omit the constants and the terms only depending on $\matr{L}$ assumed to be fixed, as
    \begin{align}
        \norm{ \matr{L}^{1/2}(\matr{X}-\pmb{\mu}(\vec{\theta}))\matr{\Sigma}_V^{-1/2}}_F^2 + n \log(|\matr{\Sigma}_V|) &= \Tr{((\matr{X}-\pmb{\mu}(\vec{\theta}))'\matr{L} (\matr{X}-\pmb{\mu}(\vec{\theta}))\matr{\Sigma}_V^{-1})} + n \log(|\matr{\Sigma}_V|) \nonumber\\
   &=   \md2(\matr{X})    + n \log(|\matr{\Sigma}_V|) \label{eq:ML}
    \end{align}
Taking the derivative in $\matr{\Sigma}_V^{-1}$ and using \cite{seber2009multivariate} for both terms, the estimator equation for $\matr{\Sigma}_V$ is given by:
$$(\matr{X}-\pmb{\mu}(\vec{\theta}))'\matr{L} (\matr{X}-\pmb{\mu}(\vec{\theta})) - n \matr{\Sigma}_V = 0.$$ 
Furthermore, expanding the negative log-likelihood, we can write it as 
\begin{align*}
\Tr{(\matr{X}'\matr{L}\matr{X}\matr{\Sigma}_V^{-1})} -2\Tr{(\matr{X}'\matr{L}\pmb{\mu}(\vec{\theta})\matr{\Sigma}_V^{-1})} + \Tr{(\pmb{\mu}(\vec{\theta})'\matr{L}\pmb{\mu}(\vec{\theta})\matr{\Sigma}_V^{-1})}  + n \log(|\matr{\Sigma}_V|).     
\end{align*}
For any differentiable function $g$ we can write the derivative of the composition $\vec{\theta}\mapsto g(\pmb{\mu}(\vec{\theta}))$ as
\begin{equation} \label{eq:compositionDer}
    \frac{\partial g(\pmb{\mu}(\vec{\theta}))}{\partial \theta_{ml}} = \sum_{i,k} \frac{\partial g(\pmb{\mu})}{\partial \mu_{ik}} \frac{\partial {\mu}(\vec{\theta})_{ik}}{\partial \theta_{ml}}.
\end{equation}
 By matrix calculus, see \cite{petersen2008matrix}, we have 
$\frac{\partial }{\partial \mu_{ik}} (-2 Tr{(\matr{X}'\matr{L}\pmb{\mu}\matr{\Sigma}_V^{-1})}) = -2(\matr{L}\matr{X}\matr{\Sigma}_V^{-1})_{ik}$ and $\frac{\partial }{\partial \mu_{ik}} Tr{(\pmb{\mu}'\matr{L}\pmb{\mu}\matr{\Sigma}_V^{-1})} = 2(\matr{L}\pmb{\mu}\matr{\Sigma}_V^{-1})_{ik}$. Plugging the latter two into  (\ref{eq:compositionDer}) with $g$ being the expanded negative log-likelihood we have for $m=1,\ldots,q$ and $l=1,\ldots,\tilde{p}$:
\begin{align*}
    \frac{\partial g(\pmb{\mu}(\vec{\theta}))}{\partial \theta_{ml}} = \sum_{i=1}^n \sum_{k=1}^p (-2(\matr{L}\matr{X}\matr{\Sigma}_V^{-1}) + 2 \matr{L}\pmb{\mu}\matr{\Sigma}_V^{-1})_{ik}  \frac{\partial {\mu}(\vec{\theta})_{ik}}{\partial \theta_{ml}} = 0
\end{align*}
which yields the estimating equation for $\pmb{\mu}$.
If $\pmb{\mu}({\pmb{\theta}}) = \matr{Z}\pmb{\theta} $, then
\begin{align*}
    \frac{\partial {\mu}({\pmb{\theta}})_{ik}}{\partial \theta_{ml}}  =\frac{\partial}{\partial \theta_{ml}} \left(\sum_{m'} z_{im'} \theta_{m'k}\right) = z_{im} \delta_{kl}
\end{align*}
where $\delta_{kl}$ is the Kronecker delta being one if $k = l$ and zero otherwise. All in all, we have for $m=1,\ldots,q$: 
\begin{align*}
    \sum_{i=1}^n (\matr{L}(\pmb{\mu}(\vec{\theta})-\matr{X})\matr{\Sigma}_V^{-1})_{il}\,  z_{im}  = 0
\end{align*}
which can be written as $\matr{Z}'\matr{L}(\matr{Z} \pmb{\theta}-\matr{X})\matr{\Sigma}_V^{-1} = \matr{0} $, and gives the desired result. Lastly, the result for $\matr{Z} = \vec{1}_n$ follows directly from $\vec{1}_n'\matr{L} = \vec{0}$.
\end{proof}

\end{appendices}
 
\bibliographystyle{abbrv} 
\bibliography{CellBib}

\end{document}